\newif\ifjournalformat
\journalformattrue

\newif\ifshowrevisions
\showrevisionsfalse

\ifjournalformat
    \documentclass[10pt,journal,letterpaper]{IEEEtran}
    \newenvironment{widefigure}{\begin{figure*}[t]}{\end{figure*}}
\else
    \documentclass{article}
    \usepackage{mathpazo}  % Palatino font with math support
    \newenvironment{widefigure}{\begin{figure}[t]}{\end{figure}}
    \newenvironment{IEEEkeywords}{\noindent\textbf{Keywords:} }{}
    
\fi

\usepackage{cite}
\usepackage{amsmath,graphicx}
\usepackage{amssymb}
\usepackage{amsfonts}
\usepackage[english]{babel}
\usepackage{xcolor,latexsym,amsthm}
\usepackage[hidelinks]{hyperref}
\usepackage[mathscr]{euscript}
\usepackage{tikz}
\usetikzlibrary{decorations.pathreplacing,decorations.pathmorphing,math}
\usepackage{placeins}

\ifshowrevisions
    \colorlet{revisionred}{red}
    \colorlet{revisionblue}{blue}
\else
    \colorlet{revisionred}{black}
    \colorlet{revisionblue}{black}
\fi

\newcommand{\E}{\mathbf{E}}

\newcommand{\Var}{\mathbf{Var}}

\newcommand{\Prob}{\mathbf{P}}

\newcommand{\Ind}{\mathbf{1}}
\newcommand{\R}{\mathbb{R}}

\newcommand{\Measure}{{\Lambda}}
\def\fade{{F}} %%% symbol for general fading random variables

\newcommand{\righthalf}{\mathcal{R}}
\newcommand{\lefthalf}{\mathcal{L}}

\newtheorem{Theorem}{Theorem}[section]
\newtheorem{Corollary}[Theorem]{Corollary}
\newtheorem{Lemma}[Theorem]{Lemma}
\newtheorem{Proposition}[Theorem]{Proposition}
\theoremstyle{definition}

\newtheorem{Remark}[Theorem]{Remark}

\begin{document}

\title{Wireless Signal Point Processes with a Reflecting Wall}
\hypersetup{
    pdftitle={Wireless Signal Point Processes with a Reflecting Wall},
    pdfauthor={H. Paul Keeler and George C. Alexandropoulos}
}

%authors
\author{{H. Paul Keeler and George~C.~Alexandropoulos}
\thanks{H.~P.~Keeler is with the School of Mathematics and Statistics, University of Melbourne, Parkville, VIC 3010, Australia (e-mail: hpkeeler@unimelb.edu.au).}

\thanks{G.~C.~Alexandropoulos is with the Department of Informatics and Telecommunications, National and Kapodistrian University of Athens, Panepistimiopolis Ilissia, 16122 Athens, Greece. (e-mail: alexandg@di.uoa.gr)}
}

\maketitle
%%%%%%%%%%%%%%%%%%%%%%%%%%%%%%%%%%%%%%%%%%%%%%%%%%%%%%%%%%%%%%%%%%%%%%%%%%%%%%%%%%%%%%%%%%%%%%

%%%%%%%%%%%%%%%%%%%%%%%%%%%%%%%%%%%%%%%%%%%%%%%%%%%%%%%%%%%%%%%%%%%%%%%%%%%%%%%%%%%%%%%%%%%%%%
%%%  Abstract
%%%%%%%%%%%%%%%%%%%%%%%%%%%%%%%%%%%%%%%%%%%%%%%%%%%%%%%%%%%%%%%%%%%%%%%%%%%%%%%%%%%%%%%%%%%%%%
\begin{abstract}
A long reflecting wall gives every wireless transmitter a mirror image, pairing its direct and reflected paths. Motivated by intelligent surfaces, we study the resulting signal point processes under two treatments: keeping the powers apart and superposing their point processes, or adding the powers transmitter by transmitter to form a recombined process. A threshold on inverse received power selects a region of transmitter locations, so Campbell's theorem reduces the mean measures to area calculations without assuming Poisson transmitters. For power-law path loss and paired fading marks, we derive the ratio of the recombined and superposed mean measures at large thresholds. The limit is independent of transmitter density and wall distance, but generally depends on the joint fading distribution. Under common fading, that dependence disappears and the limit depends only on the reflection coefficient and path-loss exponent. For Poisson transmitters, we derive paired count fluctuations and convergence of the count ratio. We also calculate the probability that the strongest individual path is direct. When the wall passes through the receiver, this probability has a closed form and the strongest received power has an exact Fr\'echet distribution. These results provide a reference model for wireless networks with reflecting surfaces.
\end{abstract}
%%%%%%%%%%%%%%%%%%%%%%%%%%%%%%%%%%%%%%%%%%%%%%%%%%%%%%%%%%%%%%%%%%%%%%%%%%%%%%%%%%%%%%%%%%%%%%

%%%%%%%%%%%%%%%%%%%%%%%%%%%%%%%%%%%%%%%%%%%%%%%%%%%%%%%%%%%%%%%%%%%%%%%%%%%%%%%%%%%%%%%%%%%%%%
%%%  Keywords
%%%%%%%%%%%%%%%%%%%%%%%%%%%%%%%%%%%%%%%%%%%%%%%%%%%%%%%%%%%%%%%%%%%%%%%%%%%%%%%%%%%%%%%%%%%%%%
\begin{IEEEkeywords}
Poisson point process, method of images, received signal power, stochastic geometry, reconfigurable intelligent surfaces.
\end{IEEEkeywords}
%%%%%%%%%%%%%%%%%%%%%%%%%%%%%%%%%%%%%%%%%%%%%%%%%%%%%%%%%%%%%%%%%%%%%%%%%%%%%%%%%%%%%%%%%%%%%%

%%%%%%%%%%%%%%%%%%%%%%%%%%%%%%%%%%%%%%%%%%%%%%%%%%%%%%%%%%%%%%%%%%%%%%%%%%%%%%%%%%%%%%%%%%%%%%
%%%  Introduction
%%%%%%%%%%%%%%%%%%%%%%%%%%%%%%%%%%%%%%%%%%%%%%%%%%%%%%%%%%%%%%%%%%%%%%%%%%%%%%%%%%%%%%%%%%%%%%
\section{Introduction}\label{sec:intro}
Someone outside makes a phone call while standing beside a long wall of a building. A signal from a nearby network base station can reach the phone. The phone's receiver may get two copies of the same transmission: one directly through the air and one after reflection from the wall. The law of reflection makes the angle of incidence equal to the angle of reflection. The method of images then turns the reflected path into a simple mirror picture. Reflect the base station across the wall, and the bent path becomes a straight path from a virtual transmitter. Now replace that one base station by a city full of randomly placed transmitters. Every transmitter acquires a paired image. This paper asks what point process of received signals the wall creates.

The related wave-level model of Keeler~\cite{keeler2026reflected} follows coherent reflections from one randomly placed transmitter between two walls. Here we fix one wall, add received powers rather than coherent fields, and study a random field of transmitters. Our wall is an infinite flat specular reflector with a constant power reflection coefficient, not a programmable surface with phase control or beam steering. Although idealised, it lets us isolate the dependence created when each real transmitter has one mirror image.

We assume there are two distinct ways to treat each pair of signals. Keeping the paths separate superposes the direct and reflected signal point processes; adding their powers before mapping produces one recombined point per transmitter. A threshold on inverse received power selects a region of transmitter locations. Campbell's theorem turns this geometric calculation into a mean measure. The wall makes the geometry different from that of a disk in the full plane. Recombination leads to a Cassinian-type threshold region with a curved boundary. Section~\ref{sec:One_Wall} derives the direct and reflected mean measures assuming only constant transmitter density. Section~\ref{s.both} then compares the two treatments of each path pair and derives bounds and asymptotics. Section~\ref{s.poisson} assumes Poisson transmitters to describe paired count fluctuations, the observed count ratio, and the strongest individual path.

Our main asymptotic result concerns the ratio of the recombined mean measure to the superposed mean measure obtained when the paths are kept separate. Throughout, $\Lambda$ with a suitable subscript denotes a mean measure. Under power-law path loss and with general paired fading marks, the ratio's limit can depend on their joint distribution. Under common fading, this dependence disappears, and the ratio becomes
\begin{equation}\label{e.ratio_intro}
\frac{\Measure_{+}(t)}{\Measure_{\cup}(t)}
\longrightarrow
\frac{(1+\kappa)^{2/\beta}}{1+\kappa^{2/\beta}},
\qquad t\to\infty.
\end{equation}
The reflection coefficient $\kappa$ and path-loss exponent $\beta$ appear in this ratio. The threshold $t$ applies to inverse received power, so increasing it admits progressively weaker signals. The ratio compares two counting rules, not two powers: a limit below one does not mean that recombination reduces received power. At finite thresholds, a reflected path can be too weak to qualify on its own yet still push its paired recombined signal across the threshold. The transmitter density $\lambda$ cancels from the ratio at every threshold, and the wall distance $a$ drops out in the limit. Under common fading, the fading distribution also disappears. The Poisson results complement this geometric limit by describing random counts, identifying whether a direct or reflected path supplies the strongest signal, and giving the distribution of the strongest power when the wall passes through the receiver.

\subsection{Related work}\label{sec:lit}
Stochastic geometry has long been used to model wireless networks with randomly placed transmitters and receivers~\cite{baccelli1997stochastic,baccelli2001coverage,andrews2011tractable,BB09,blaszczyszyn2020coverage}. Models of intelligent surfaces likewise randomise base stations, users, blockages, or reflectors, often to study coverage and rate. Kishk and Alouini~\cite{kishk2020exploiting} place surfaces on random blockages; Zhang et al.~\cite{zhang2021reconfigurable} use a Poisson cluster model for multicell networks; and Shafique et al.~\cite{shafique2022stochastic} calculate aggregate interference with many surfaces. Pradhan, Imani, and Dhillon~\cite{pradhan2025beamshaping} instead use a physically consistent model that includes mutual coupling and structural scattering.

Several models explicitly retain dependence in the reflecting infrastructure. Shi et al.~\cite{shi2023coverage} associate a fixed number of surfaces with each base station. Lin and Yu~\cite{lin2025ternary} use the triangle formed by a base station, surface, and user as their basic geometric unit. Choi et al.~\cite{choi2025stochastic} place roadside surfaces and vehicle users on the same random road system, while Sun et al.~\cite{sun2026stochastic} condition reflecting surfaces on their associated Poisson base stations and study the resulting aggregate reflected signal by shot noise and Laplace transforms.

The reflection geometry itself is considered by Di Renzo and Song~\cite{direnzo2019reflection}, who model random finite metasurface-coated objects, and by Narayanan et al.~\cite{narayanan2017coverage}, who mirror Poisson base stations in random finite reflectors but retain only the shortest reflected path. Our setting is complementary: a single fixed wall gives every transmitter an exact, deterministic mirror image. We keep the associated direct and reflected powers paired, whether the paths are superposed or their powers recombined.

Another line of work maps transmitters to their inverse received powers, forming a propagation or projection process; B{\l}aszczyszyn et al.~\cite[Chapter~4]{book2018stochastic} review this construction. Under suitable propagation effects, sufficiently rare strong signals can appear approximately Poisson even when the transmitter locations are not Poisson~\cite[Sec.~IV-A]{keeler2016stronger}, \cite[Theorem~2.2 and Remark~2.4]{keeler2018wireless}. This is distinct from our exact Poisson mapping results when the transmitters themselves are Poisson. Extremal shot noise, which takes a maximum rather than a sum, is also established theory~\cite[Vol.~I, Sec.~2.4]{BB09}. Nguyen and Baccelli~\cite{nguyen2010best} derived the distribution of the maximum received signal strength in a Poisson wireless network, together with its joint distribution with the interference, and used these quantities to study the best signal quality. Our strongest-path calculation asks a more particular question: which of the paired direct and reflected paths attains that maximum?

For reference, Table~\ref{tab:notation} collects the main symbols used for the two signal models.
\begin{table}[t]
\centering
\caption{Main notation for the one-wall model.}
\label{tab:notation}
\footnotesize
\renewcommand{\arraystretch}{1.12}
\begin{tabular}{@{}p{0.38\columnwidth}p{0.55\columnwidth}@{}}
\hline
\textbf{Symbols} & \textbf{Meaning} \\
\hline
$o,a,\kappa,\beta$ & Receiver, wall distance, reflection coefficient, path-loss exponent \\
$r,\hat r_1$ & Direct distance and image distance, equal to the reflected-ray length \\
$\Phi,\Phi'$ & Real transmitters and their virtual images \\
$\fade_i,\fade_i'$ & Direct and reflected fading marks \\
$L_i,L_i',L_i^{+}$ & Direct, reflected and recombined inverse powers \\
$\Xi,\Xi'$ & Direct and reflected signal point processes \\
$\Xi_{\cup},\Xi_{+}$ & Superposed (paths kept separate) and recombined signal processes \\
$\Measure_{\lefthalf},\Measure_{\righthalf}$ & Direct and reflected mean measures \\
$\Measure_{\cup},\Measure_{+}$ & Corresponding superposed and recombined mean measures \\
\hline
\end{tabular}
\end{table}

%%%%%%%%%%%%%%%%%%%%%%%%%%%%%%%%%%%%%%%%%%%%%%%%%%%%%%%%%%%%%%%%%%%%%%%%%%%%%%%%%%%%%%%%%%%%%%

%\section{Geometrical Signal Model}\label{sec:Geometrical_Model}
\section{Geometric model and mean measures}\label{sec:One_Wall}
Place the receiver at the origin $o$ beside an infinitely long vertical wall at $x=a$, as shown in Figure~\ref{fig:onewall}. A transmitter at the point $z=(x,y)\in\R^2$ lies on the same side of the wall as the receiver, so its horizontal coordinate satisfies $x\leq a$. The wall distance $a\geq0$ is the perpendicular distance from the receiver to the wall. We also write the transmitter position in polar coordinates $(r,\theta)$.

We use a constant power reflection coefficient $\kappa\in[0,1]$: the wall reflects the fraction $\kappa$ of the incident power and absorbs the rest. A lossless wall has reflection coefficient $\kappa=1$. The constant coefficient is an idealization. Real reflection depends on wall material, incidence angle, frequency, and polarization. We use this simple specular-reflection model to keep the geometric pairing visible. The wave-level model in the paper~\cite{keeler2026reflected} discusses these physical limitations in more detail.

%%%%%%%%%%%%%%%%%%%%%%%%%%
%  Figure  -- One walls
%%%%%%%%%%%%%%%%%%%%%%%%%%
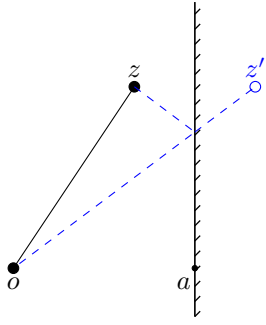
\begin{figure}[t]
\centering
% Figure source: fig_transmitter_image.tex
% Figure 1: original two-panel drawing combined into one panel.
% Preserve line weights, wall hatching and point symbol sizes;
% the reflected contact is corrected to (5,2.25).
% Local TikZ style for the hatched reflecting wall.
\tikzset{
  interface_right/.style={
    postaction={draw,decorate,decoration={border,angle=-45,
      amplitude=0.1cm,segment length=2mm}}}
}
\begin{tikzpicture}[scale=0.8]
\filldraw
    (2,0) circle (2.5pt) node[align=center,below] {$o$}
    (5,0) circle (1.25pt) node[align=right,below] {\hspace{-1em}$\,a$}
    (4,3) circle (2.5pt) node[above] {$z$};
% Original hatched wall, exactly as in Figure 2.
\draw[-] (5,-.8) -- (5,4.4);
\draw[line width=.5pt,interface_right] (5,-.8) -- (5,4.4);
% Direct path.
\draw (2,0) -- (4,3);
% Physical specular path: incidence, then reflection to the receiver.
\draw[dashed][color=blue] (4,3) -- (5,2.25) -- (2,0);
% Virtual transmitter and the straight continuation beyond the wall.
\draw[dashed][color=blue] (5,2.25) -- (6,3);
\draw[color=blue] (6,3) circle (2.5pt) node[above] {$z'$};
\end{tikzpicture}
\caption{A transmitter $z$ and its virtual image $z^\prime$ across a reflecting wall. The solid line is the direct ray to the receiver $o$. The blue dashed segments show the incident ray, the reflected ray, and their straight extension beyond the wall to the virtual transmitter. The reflected path has the same length as the straight path from image $z^\prime$ to receiver $o$.}\label{fig:onewall}
\end{figure}

We assume that the path-loss function $\ell(r)$ is positive and non-increasing. A standard single-slope model is $\ell(r)=r^{-\beta}$ with path-loss exponent $\beta>0$.

For a single transmitter at location $z$, the direct and reflected path powers, before fading marks are introduced, are $\ell(r)$ and $\kappa\ell(\hat r_1)$, respectively. When these powers are recombined noncoherently or averaged over their relative phase, the received power is
\begin{align}
\hat{S}(z)=&\ell(r)+\kappa\,\ell\left(\hat{r}_1\right)
\,,
\label{e.signal1}
\end{align}
where the image distance $\hat r_{1}$ is equal to the reflected-ray length. Expression~\eqref{e.signal1} is not the instantaneous coherent power of two phase-locked paths, which would contain a phase-dependent cross term. The model is appropriate, for example, when the paths are resolved and their powers are combined noncoherently, or when phase averaging removes the cross term. 

The line-of-sight signal travels from transmitter $z$ to receiver $o$. In polar coordinates, the horizontal and vertical coordinates of transmitter $z$ are $x=r\cos\theta$ and $y=r\sin\theta$, respectively. The distance from the transmitter to the receiver via the wall is then
\begin{equation}\label{e.rhat}
    \hat{r}_{1} =\sqrt{(2a-x)^2+y^2}\,.
\end{equation}
We record a simple geometric fact that we will use repeatedly. For every transmitter location $z\in\lefthalf$,
\begin{equation}\label{e.rhatgeq}
\hat{r}_1\geq r\,,
\end{equation}
because the squared-distance difference $\hat{r}_1^2-r^2=(2a-x)^2-x^2=4a(a-x)\geq 0$ when the coordinate $x\leq a$ and the wall distance $a\geq 0$. Equivalently, the image distance equals the length of the reflected ray path, which the triangle inequality bounds below by the direct distance. In words, the reflected path is never shorter than the direct path.

\subsection{Random transmitter field}\label{s.transmitters}
We now extend our model to incorporate spatial randomness by considering multiple transmitters distributed according to a spatial point process. Specifically, on the half-plane $\lefthalf=\{(x,y): x\leq a ; x,y\in \R\}$ to the left of the wall, we model transmitter locations as a point process $\Phi=\{Z_i\}_{i=1}^{\infty}$ with constant density $\lambda>0$. Thus, for any bounded planar region $B\subseteq\lefthalf$, the expected number of transmitters in region $B$ is $\lambda |B|$, where the quantity $|B|$ denotes its area. At this stage we do not assume the Poisson property. The results in this section hold for any point process with constant density, by Campbell's theorem~\cite{kingman1992poisson,book2018stochastic}. We introduce the Poisson assumption later, in Section~\ref{s.poisson}, where distributional statements first become necessary.

The key insight from the method of images is that reflections from the wall create virtual transmitter images $\Phi'=\{Z_i'\}_{i=1}^{\infty}$ on the opposite side, as shown in Figure~\ref{fig:onewallpoints}. Each virtual image $Z_i'$ is deterministically related to its corresponding real transmitter $Z_i$ through mirror reflection across the wall. Mirror reflection preserves areas, so the collection $\Phi'$ is itself a point process with constant density $\lambda$ on the right half-plane $\righthalf$. Critically, these two processes are strongly coupled: the transmitter process $\Phi$ completely determines the image process $\Phi'$ through the reflection geometry.

%%%%%%%%%%%%%%%%%%%%%%%%%%
%  Transmitter and image point fields --
%%%%%%%%%%%%%%%%%%%%%%%%%%
\begin{figure}[t]
\centering
% Figure source: fig_transmitter_image_points.tex
% Smaller geometric footprint; labels keep their LaTeX size.
% Local TikZ style for the hatched reflecting wall.
\tikzset{
  interface_right/.style={
    postaction={draw,decorate,decoration={border,angle=-45,
      amplitude=0.1cm,segment length=2mm}}}
}
\begin{tikzpicture}[scale=0.70,font=\footnotesize]
\filldraw 
%(0,0) circle (1pt) node[align=right, below] {\hspace{1em}  -$b$} --
(2,0) circle (2.4pt) node[align=center, below] {$o$}       %observer
(5,0) circle (.60pt) node[align=right,  below] {} 
(4,3) circle (2.4pt) node[above] {$Z_1$} ;  %transmitter 1
%   \draw[<->] (0,-1) --(0,5); %left wall
\draw[-] (5,-1.5) --(5,5); %right wall
\draw[line width=.5pt,interface_right] (5,-1.45) --(5,4.95); %right wall
   % \draw[dashed](2,0) -- (4,3)  ;  %transmitter to observer
\draw[color=blue] (6,3) circle (2.4pt) node[above] {$Z_1'$} ;  %virtual right node  
%\draw[color=red] (-4,3) circle (2.4pt);  %virtual left node  
    \draw[dashed][color=blue] (6,3) -- (2,0); %first right arrow A    
  %   \draw[dashed][color=red] (4,3) --(-4,3) %first right arrow A
  %        --(2,0); %first right arrow B
%%%transmitter 2
\filldraw  (1,4) circle (2.4pt) node[above] {$Z_2$} ;  %transmitter 2
\draw[color=blue] (9,4) circle (2.4pt) node[above] {$Z_2'$};  %virtual node (transmitter 2)
%\draw[dashed][color=red] (1,4) --(9,4); %first right arrow (transmitter 2)
\draw[dashed][color=blue] (2,0) --(9,4); %second right arrow (transmitter 2)
%%% transmitter 3
\filldraw (0,-1) circle (2.4pt) node[above] {$Z_3$} ;  %transmitter 3
\draw[color=blue] (10,-1) circle (2.4pt)node[above]{$Z_3'$} ;  %virtual node (transmitter 3)
 % \draw[dashed][color=green] (0,-1) --(10,-1); %first right arrow (transmitter 3)
    \draw[dashed][color=blue] (2,0) --(10,-1); %second right arrow (transmitter 3)
\end{tikzpicture}
\caption{The real transmitters $\Phi=\{Z_i\}_i$ lie in the left half-plane $\lefthalf$. Their paired images $\Phi'=\{Z_i'\}_i$ lie in the right half-plane $\righthalf$. Reflection fixes each image once its transmitter is placed.}
\label{fig:onewallpoints}
\end{figure}
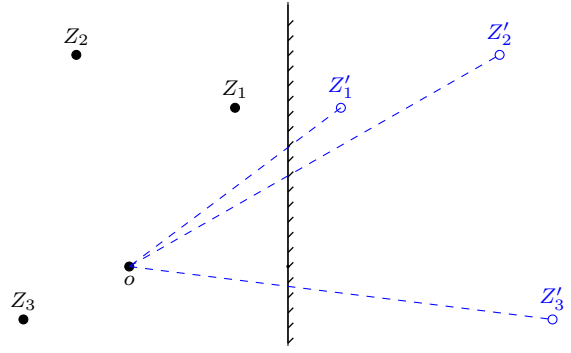

%%%%%%%%%%%%%%%%%%%%%%%%%%

To model propagation effects beyond geometric path loss, we attach a positive random mark $\fade_i$ to each transmitter $Z_i\in\Phi$. We assume the marks $\fade_i$ are independent copies of a generic positive random variable $\fade$, and that they are independent of the transmitter process $\Phi$. More general propagation models can allow dependence between marks at different transmitters~\cite[Sec.~4.1.1]{book2018stochastic}. Under our assumption, the marked transmitters form the independently marked point process $\tilde{\Phi}=\{(Z_i,\fade_i)\}_{i=1}^{\infty}$~\cite[Vol.~I, Chapter~2]{BB09}. The direct received power from transmitter $Z_i$ at the origin is $\fade_i\ell(\lVert Z_i\rVert)$. In standard wireless network models, the marks typically represent small-scale fading or shadowing.

\subsection{Direct signals}\label{s.direct}
A standard approach with wireless network models based on point processes is to study interference and coverage using transformed point processes. We do this by inverting the received signal powers. For each transmitter $Z_i \in \Phi$, we define the direct inverse received power $L_i$ by
\begin{equation}
    L_i \triangleq \frac{1}{ \ell(\lVert Z_i\rVert) \fade_i}\,, \quad Z_i \in \Phi \,
\end{equation}
where the distance $\lVert Z_i\rVert$ runs from the receiver $o$ to transmitter $Z_i$. This transformation produces a point process $\Xi  \triangleq\{L_i\}_{i=1}^{\infty}$ on the positive real line $\R^+$, variously termed the \emph{propagation process}, \emph{projection process}, or \emph{path loss with fading process} in the wireless literature. Under this representation, stronger signals correspond to points closer to zero, while weaker signals map to larger values. Although this inversion might seem counterintuitive, it proves mathematically convenient for analyzing interference and deriving coverage probabilities. For further background on this technique, see~\cite[Chapter 4]{book2018stochastic}.

For a threshold $t\geq0$, let the count $\Xi(t)$ be the number of points of the signal point process $\Xi$ in the interval $(0,t]$. We use the same convention for the reflected, superposed, and recombined signal point processes. We first derive the direct mean measure, the expected value of $\Xi(t)$. These are expected count results, which do not require Poisson transmitters. Later, in Section~\ref{s.poisson}, the Poisson mapping theorem will let us use the mean measures to determine the distributions of the mapped processes. For the direct signals, the mean measure is
\begin{align}
\Measure_{\lefthalf}(t)
&\triangleq \E\big[\Xi(t)\big] \nonumber \\
&= \E\left[\sum_{L\in \Xi} \Ind(L\leq t)\right] \nonumber \\
&= \E\left[\sum_{Z\in \Phi}
    \Ind\!\left(\frac{1}{\ell(\lVert Z\rVert)\fade}\leq t\right)\right]\,.
\label{intensity1}
\end{align}
where the indicator function is denoted by $\Ind$.

In special cases, we can derive closed-form expressions for the mean measure $\Lambda$. For the path-loss model $\ell$, we assume the form $\ell(r)=1/h(r)$, where the function $h$ is left-continuous and non-decreasing with generalized inverse $h^{-1}(y)=\inf\{x:h(x)>y\} $. We assume the finite expected area condition $\E[(h^{-1}(t\fade))^2]<\infty$ for every finite threshold $t>0$. For the power-law model, this condition reduces to $\E[\fade^{2/\beta}]<\infty$. If the transmitters $\Phi=\{Z_i\}$ were instead scattered on the entire plane $ \R^2$, then the mean measure is simply
\begin{equation}\label{e.full_plane_general}
\Measure_{\R^2}(t)=\pi\lambda \E[(h^{-1}(t\fade))^2]\,.
\end{equation}
This expression also appears in~\cite[Sec.~III-A, Eq.~(3)]{keeler2016stronger} and~\cite[Proposition~2.9, Eq.~(2.5)]{keeler2018wireless}. It shows that the measure $\Measure_{\R^2}(t)$ is simply the average number of points of $\Phi$ located in a disk of (random) radius $h^{-1}(t\fade)$, meaning we  need to calculate the expected area of the disk with radius  $h^{-1}(t\fade)$. The standard power-law path loss $\ell(r)=r^{-\beta}$ gives
\begin{equation}\label{e.full_plane_power}
\Measure_{\R^2}(t)=\pi\lambda \E[\fade^{2/\beta}] t^{2/\beta}\,.
\end{equation}

\subsubsection{Transmitters in the half-plane}
We use this disk interpretation to derive the mean measure $\Lambda$ in the setting with a reflecting wall. To do that, we first consider the disk $B_o(d)$ of radius $d$ centred at the receiver $o$, clipped (or truncated) by the wall at $x=a$. In other words, this region is the intersection $B_o(d)\cap \lefthalf$. In the area notation below, the calligraphic $\mathcal O$ records that the disk is centred at the origin. Its area is
\begin{align}
    A_{\mathcal O\cap \lefthalf}(d) &\triangleq|B_o(d)\cap \lefthalf| \nonumber\\
    &=\pi d ^2- A_{\mathcal O\cap {\righthalf}}(d)\, , \label{e.areadiskleft}
\end{align}
where the term $A_{\mathcal O\cap {\righthalf}}(d)$ is the area of the part of the disk lying on the other side of the wall, the right half-plane $\righthalf\triangleq\{(x,y): x>a ; x,y\in \R\}$. The area $A_{\mathcal O\cap {\righthalf}}(d)$ is calculated using basic trigonometry or calculus, giving
\begin{equation}\label{e.areadisktruc}
        A_{\mathcal O\cap {\righthalf}}(d)=
    \begin{cases}
      d^2\arccos(a/d)-a\sqrt{d^2-a^2}, & d>a, \\
      0, & d\leq a.
    \end{cases}
\end{equation}
Using this closed-form expression, we can state the mean measure of the point process formed from the inverse received powers. We stress that the next result requires no Poisson assumption. It holds for any transmitter point process with constant density $\lambda$ and positive marks that are independent, share a common distribution, and are independent of the transmitter process, by Campbell's theorem.
\begin{Proposition}\label{p.los}
Let the point process $\Phi$ on $\lefthalf$ have constant density $\lambda>0$ and positive marks $\{\fade_i\}$ that are independent of each other and of $\Phi$, and share a common distribution. The inverse received powers $\Xi=\{L_i\}_{i=1}^{\infty}$  form a point process on the real line $\R^+$ with mean measure
\begin{equation}\label{e.mean_direct_general}
\Measure_{\lefthalf}(t)=\lambda \E[A_{\mathcal O\cap \lefthalf}(h^{-1} (t\fade))]\, ,
\end{equation}
where the area $A_{\mathcal O\cap \lefthalf}$ is given by expression~\eqref{e.areadiskleft}.
If the fading mark $\fade$ is a constant $c>0$, then the mean measure simplifies to
\begin{equation}
\Measure_{\lefthalf}(t)=\lambda A_{\mathcal O\cap \lefthalf}(h^{-1} (t c))\, .
\end{equation}
\end{Proposition}
\begin{proof}
The event $\{L\leq t\}$ is equivalent to the event $\{h(\lVert Z\rVert)\leq t\fade\}$, hence to $\{\lVert Z\rVert \leq h^{-1}(t\fade)\}$. Conditioning on $\fade$ and applying Campbell's theorem to the count of points of $\Phi$ in the region $B_o(h^{-1}(t\fade))\cap \lefthalf$ gives the result.
\end{proof}

For the power law path model $\ell(r)=r^{-\beta}$, we have $h(r)=r^{\beta}$ and $h^{-1}(y)=y^{1/\beta}$, so the expression is simply
\begin{equation}\label{e.intensLOSpower}
\Measure_{\lefthalf}(t)=\lambda \E[A_{\mathcal O\cap \lefthalf}((t\fade)^{1/\beta})]\, .
\end{equation}
For the constant fading mark $\fade=1$, the mean measure is
\begin{equation}
\Measure_{\lefthalf}(t)=\lambda A_{\mathcal O\cap \lefthalf}(t^{1/\beta}) \, .
\end{equation}
When the wall passes through the receiver, the relevant region is exactly a half-disk. For a fixed radius, moving the wall away increases that area until the whole disk lies on the transmitter side.

\subsection{Reflected signals and virtual transmitters}\label{s.virtual}
Reflection gives the image point process $\Phi'=\{Z_i'\}$ the same constant density $\lambda$ on the right half-plane $\righthalf$. This follows because reflection preserves areas and does not require a Poisson assumption. For each image $Z_i'$, the reflected fading mark is $\fade_i'>0$; the reflection coefficient $\kappa$ is kept separate from this mark throughout. We assume that the mark pairs $(\fade_i,\fade_i')$ are independent across transmitters, have the same joint distribution, and are independent of the transmitter process $\Phi$. The generic reflected mark is $\fade'$, for which we also assume the finite expected area condition. Thus the reflected received power is $\kappa\fade_i'\ell(\lVert Z_i'\rVert)$. The fading marks on the two paths of a transmitter may be dependent. When a result needs common fading, we state the mark equality $\fade_i'=\fade_i$ explicitly. For reflection coefficient $\kappa=0$, we take the reflected process $\Xi'$ to be empty and set the reflected mean measure $\Measure_{\righthalf}(t)=0$. For a positive reflection coefficient $\kappa>0$, define the reflected inverse received power $L_i'$ by
\begin{equation}
    L_i' \triangleq \frac{1}{\kappa\, \ell(\lVert Z_i'\rVert) \fade_i'}\,, \quad Z_i' \in \Phi' \, .
\end{equation}
The reflected inverse powers form the point process $\Xi'=\{L_i'\}$ on $\R^+$. Campbell's theorem and the area $A_{\mathcal O\cap \righthalf}(d)$ in expression~\eqref{e.areadisktruc} give the reflected counterpart of Proposition~\ref{p.los}.

\begin{Proposition}\label{p.nlos}
For a transmitter process with constant density $\lambda$ and the mark assumptions above, the inverse reflected received powers $\Xi'=\{L_i'\}_{i=1}^{\infty}$ have, for a positive reflection coefficient $\kappa>0$, the mean measure
\begin{equation}\label{e.mean_reflected_general}
\Measure_{\righthalf}(t)=\lambda \E\big[A_{\mathcal O\cap \righthalf}(h^{-1} (t\kappa\fade'))\big] \,.
\end{equation}
If the reflected fading mark $\fade'$ is a constant $c>0$, this simplifies to
\begin{equation}
\Measure_{\righthalf}(t)=\lambda A_{\mathcal O\cap \righthalf}(h^{-1} (t\kappa c))\,.
\end{equation}
In particular, unit fading gives the mean measure $\Measure_{\righthalf}(t)=\lambda A_{\mathcal O\cap \righthalf}(h^{-1}(t\kappa))$. For reflection coefficient $\kappa=0$, the mean measure $\Measure_{\righthalf}(t)=0$ by convention.
\end{Proposition}

Each direct inverse power remains paired with the reflected inverse power from the same transmitter. We therefore do not treat the two propagation processes as independent. For example, when the wall distance is $a=0$, the reflection coefficient is $\kappa=1$, and common fading is used, the paired values satisfy $L_i'=L_i$.

\section{Separate and recombined signals}\label{s.both}

We treat the direct and reflected paths in two distinct ways. In the first, we keep their received powers separate and superpose the corresponding signal point processes. In the second, we add the two powers from each transmitter and record one recombined signal. Figure~\ref{g.pathcounting} shows the difference on the line of inverse received powers. The two points associated with one transmitter can both lie beyond the threshold, while their recombined inverse power falls below it.

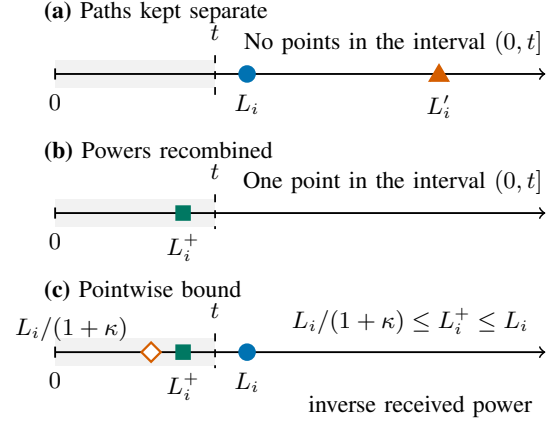
\begin{figure}[t]
\centering
% Figure source: fig_inverse_power_counting.tex
% TikZ source for Figure 3 (inverse received-power counting).
% Kept separate from the main article source by design.
% Native TikZ geometry at approximately 80 percent of the former printed width.
% Same inverse received-power scale on all three rows. Locations are proportional to
% L_i/t=1.2, L_i'/t=2.4, L_i^+/t=0.8 and, for kappa=1, L_i/((1+kappa)t)=0.6.
\begingroup
\definecolor{countBlue}{RGB}{0,114,178}
\definecolor{countOrange}{RGB}{213,94,0}
\definecolor{countGreen}{RGB}{0,120,98}
\begin{tikzpicture}[scale=0.92,x=1cm,y=1cm,font=\small,
    draw=black,text=black,line width=0.7pt]
\path[use as bounding box] (0,-4.90) rectangle (7.65,1.0);
% Counted intervals. Fill and line style are independent of revision colours.
\fill[black!5] (0.4,-0.2) rectangle (2.7,0.2);
\fill[black!5] (0.4,-2.2) rectangle (2.7,-1.8);
\fill[black!5] (0.4,-4.2) rectangle (2.7,-3.8);
\draw[dashed,line width=0.65pt] (2.7,0.27) -- (2.7,-0.3);
\draw[dashed,line width=0.65pt] (2.7,-1.73) -- (2.7,-2.23);
\draw[dashed,line width=0.65pt] (2.7,-3.73) -- (2.7,-4.23);
\node[anchor=south] at (2.7,0.35) {$t$};
\node[anchor=south] at (2.7,-1.65) {$t$};
\node[anchor=south] at (2.7,-3.65) {$t$};
\node[anchor=west] at (0.1,0.88) {\textbf{(a)} Paths kept separate};
\node[anchor=west] at (0.1,-1.12) {\textbf{(b)} Powers recombined};
\node[anchor=west] at (0.1,-3.12) {\textbf{(c)} Pointwise bound};
\draw[->] (0.4,0) -- (7.45,0);
\draw[->] (0.4,-2) -- (7.45,-2);
\draw[->] (0.4,-4) -- (7.45,-4);
\draw (0.4,-0.08) -- (0.4,0.08);
\draw (0.4,-2.08) -- (0.4,-1.92);
\draw (0.4,-4.08) -- (0.4,-3.92);
\node[anchor=north] at (0.4,-0.16) {$0$};
\node[anchor=north] at (0.4,-2.16) {$0$};
\node[anchor=north] at (0.4,-4.16) {$0$};
% Circle: direct. Triangle: reflected. Square: recombined. Diamond: lower bound.
\filldraw[draw=countBlue,fill=countBlue] (3.16,0) circle[radius=3.2pt];
\filldraw[draw=countOrange,fill=countOrange]
    (5.92,0.13) -- (5.785,-0.10) -- (6.055,-0.10) -- cycle;
\node[anchor=north] at (3.16,-0.16) {$L_i$};
\node[anchor=north] at (5.92,-0.16) {$L_i'$};
\filldraw[draw=countGreen,fill=countGreen]
    (2.142,-2.098) rectangle (2.338,-1.902);
\node[anchor=north] at (2.24,-2.16) {$L_i^{+}$};
\node[anchor=east] at (7.55,0.45) {No points in the interval $(0,t]$};
\node[anchor=east] at (7.55,-1.55) {One point in the interval $(0,t]$};
% Panel (c): L_i/(1+kappa) = 0.6t, L_i^+=0.8t, L_i=1.2t for kappa=1.
\draw[countOrange,fill=white,line width=0.85pt,rotate around={45:(1.78,-4)}]
    (1.682,-4.098) rectangle (1.878,-3.902);
\filldraw[draw=countGreen,fill=countGreen]
    (2.142,-4.098) rectangle (2.338,-3.902);
\filldraw[draw=countBlue,fill=countBlue] (3.16,-4) circle[radius=3.2pt];
\node[anchor=east] at (1.59,-3.67) {$L_i/(1+\kappa)$};
\node[anchor=north] at (2.24,-4.16) {$L_i^{+}$};
\node[anchor=north] at (3.16,-4.16) {$L_i$};
\node[anchor=east] at (7.35,-3.55) {$L_i/(1+\kappa)\leq L_i^{+}\leq L_i$};
\node[anchor=north east] at (7.40,-4.48) {inverse received power};
\end{tikzpicture}
\endgroup
\caption{Three views of the inverse received powers from one transmitter on the same scale. The shaded intervals mark the interval $(0,t]$. (a) Both paths fall to the right of the threshold $t$, so neither is counted. (b) In this particular example, recombination moves the inverse power to the left of the threshold, so one transmitter is counted. (c) The simpler pointwise bound of Proposition~\ref{p.sandwich}, valid under the mark inequality $\fade_i'\leq\fade_i$, places the recombined inverse power between $L_i/(1+\kappa)$ and $L_i$. The numerical example uses reflection coefficient $\kappa=1$, direct inverse received power $L_i=1.2t$, and reflected inverse received power $L_i'=2.4t$. Hence recombined inverse power $L_i^{+}=0.8t$ and lower bound $L_i/(1+\kappa)=0.6t$.}
\label{g.pathcounting}
\end{figure}

\subsection{Signals kept separate}\label{s.separate}
Keeping the paths separate amounts to superposing their two signal point processes. We denote the resulting process by
\begin{equation}\label{e.signal_union}
\Xi_{\cup}\triangleq \Xi \cup \Xi'\,.
\end{equation}
The union symbol $\cup$ denotes point process superposition, interpreted as a multiset union so that coincident points retain their multiplicity. It acts on the two point processes: both inverse power points from a transmitter remain present. Each point can also be marked as direct or reflected. The mean measure $\Measure_{\cup}(t)$ of the superposed signal process $\Xi_{\cup}$ is
\begin{equation}\label{e.intenssep}
\Measure_{\cup}(t)=\Measure_{\lefthalf}(t)+\Measure_{\righthalf}(t)\,.
\end{equation}
This identity follows from the linearity of expectation: point counts add under superposition even when the two point processes are dependent. It holds for the general transmitter process considered here and does not require the Poisson assumption. The quantity $\Measure_{\cup}$ should not be confused with the full-plane mean measure $\Measure_{\R^2}$, the mean measure of a field of real transmitters filling the entire plane with no wall. For a lossless wall with reflection coefficient $\kappa=1$ and equally distributed path marks, $\fade'\overset{d}{=}\fade$, we have $A_{\mathcal O\cap \lefthalf}(d)+A_{\mathcal O\cap \righthalf}(d)=\pi d^2$, and therefore
\begin{equation}\label{e.fullplane}
\Measure_{\cup}(t)=\Measure_{\R^2}(t)\,.
\end{equation}
Thus a lossless wall unfolds the occupied half-plane onto the full plane at the level of mean measure. For a lossy wall with reflection coefficient $0\leq\kappa<1$ and equally distributed path marks, attenuation can only reduce the reflected contribution, giving $\Measure_{\cup}(t)\leq\Measure_{\R^2}(t)$. For example, with unit fading and power-law path loss, equality holds when the threshold radius $t^{1/\beta}\leq a$. The direct threshold disk then lies entirely to the left of the wall, unlike the clipped disk in Figure~\ref{g.thresholdregions}(a), and no reflected path qualifies.

Under common fading, a reflected path cannot be stronger than its paired direct path. Indeed, the geometric inequality~\eqref{e.rhatgeq} and the bound $\kappa\leq1$ on the reflection coefficient give
\begin{equation}\label{e.path_power_order}
\kappa\fade_i\ell(\hat r_{1,i})\leq \fade_i\ell(r_i).
\end{equation}
Hence the largest path power is attained by a direct path. This pointwise observation does not require Poisson transmitters; the corresponding distribution of the smallest inverse received power is derived in Section~\ref{s.poisson}, where the Poisson assumption is imposed.

\subsection{Signals recombined}\label{s.recombined}
This approach adds the direct and reflected received powers from each transmitter before recording one signal point. We denote the resulting inverse received power by $L_i^{+}$. The superscript $+$ records the addition of received powers; it does not mean the sum $L_i+L_i'$ of the inverse powers. For a positive reflection coefficient, $\kappa>0$, the inverse received power is the parallel sum
\begin{equation}\label{e.parallel}
\begin{aligned}
L_i^{+} &\triangleq \left(\frac{1}{L_i}+\frac{1}{L_i'}\right)^{-1} \\
&= \frac{1}{\fade_i\,\ell(r_i)+\kappa\fade_i'\,\ell(\hat{r}_{1,i})}\,.
\end{aligned}
\end{equation}
For the transmitter $Z_i=(x_i,y_i)$, the direct distance is $r_i=\lVert Z_i\rVert$, and the image distance $\hat r_{1,i}$ follows from expression~\eqref{e.rhat} with coordinates $x=x_i$ and $y=y_i$. The final equality in expression~\eqref{e.parallel} also applies for reflection coefficient $\kappa=0$, giving the identity $L_i^{+}=L_i$. Under common fading, the mark equality $\fade_i'=\fade_i$ holds. With power-law path loss $\ell(r)=r^{-\beta}$, expression~\eqref{e.parallel} becomes
\begin{equation}\label{e.parallelpower}
L_i^{+} = \frac{1}{\fade_i\left(r_i^{-\beta}+\kappa\, \hat{r}_{1,i}^{-\beta}\right)}\,,
\end{equation}
which we use for the sandwich bound and the geometric proof below. The recombined inverse powers form the signal point process $\Xi_{+}\triangleq\{L_i^{+}\}_{i=1}^{\infty}$. The subscript $+$ carries the same convention: the received powers are added before mapping to inverse power.
For a positive threshold $t>0$, Campbell's theorem gives the mean measure $\Measure_{+}(t)$ of the signal point process $\Xi_{+}$ in the same way as expression~\eqref{intensity1}, without a Poisson assumption:
\begin{equation}\label{e.intenscombdef}
\begin{aligned}
\Measure_{+}(t)
&= \E\left[\sum_{Z\in \Phi} \Ind\big(L^{+}\leq t\big)\right] \\
&= \lambda\int_{\lefthalf}
\Prob\!\left(\fade r^{-\beta}+\kappa\fade'\hat r_1^{-\beta}
    \geq\frac{1}{t}\right)\,dz\,.
\end{aligned}
\end{equation}
Here the distances $r$ and $\hat r_1$ are the direct and image distances associated with the location $z$. At threshold $t=0$, the mean measure is zero.

Unlike the superposed case in Section~\ref{s.separate}, the region in expression~\eqref{e.intenscombdef} does not reduce to a clipped disk. Under common fading, writing the defining inequality as $(r^{\beta}-t\fade)(\hat r_1^{\beta}-t\fade\kappa)\leq \kappa (t\fade)^2$ gives a Cassinian-type region. Its structure recalls the classical Cassinian oval, defined by a constant product of distances to two fixed points~\cite[Section~5.16]{lawrence1972catalog}; we evaluate its area numerically.

\subsection{Bounds and asymptotics}\label{s.bounds}
The threshold region does not reduce to the disk geometry of the separate paths. We therefore use the disk geometry to bound its mean measure and determine its growth at large thresholds. The first result allows arbitrary positive paired fading marks. When the reflected mark is no larger than its direct counterpart, it also gives the simpler sandwich bounds used in Figure~\ref{g.pathcounting}(c).

\begin{Proposition}[Sandwich bounds]\label{p.sandwich}
Suppose that the transmitter process $\Phi$ has constant density $\lambda>0$, uses power-law path loss with exponent $\beta>0$, and has positive mark pairs $(\fade_i,\fade_i')$ satisfying the assumptions of Section~\ref{s.virtual}. For every threshold $t\geq0$, the recombined mean measure satisfies
\begin{equation}\label{e.sandwich_general}
\begin{aligned}
\Measure_{\lefthalf}(t)
&\leq \Measure_{+}(t)\\
&\leq \lambda\E\!\left[A_{\mathcal O\cap\lefthalf}
\!\left(\big(t(\fade+\kappa\fade')\big)^{1/\beta}\right)\right].
\end{aligned}
\end{equation}
If the mark inequality $\fade_i'\leq\fade_i$ holds almost surely, then the upper bound simplifies to
\begin{equation}\label{e.sandwich}
\Measure_{\lefthalf}(t)\leq\Measure_{+}(t)
\leq\Measure_{\lefthalf}((1+\kappa)t).
\end{equation}
In particular, the latter bounds hold under common fading $\fade_i'=\fade_i$.
\end{Proposition}

\begin{proof}
The threshold case $t=0$ is immediate. For a positive threshold $t>0$, the geometric inequality~\eqref{e.rhatgeq} gives $\hat r_1\geq r$. For arbitrary positive mark pairs, the received powers therefore satisfy
\begin{equation}\label{e.pathwise_sandwich}
\fade r^{-\beta}\leq
\fade r^{-\beta}+\kappa\fade'\hat r_1^{-\beta}
\leq (\fade+\kappa\fade')r^{-\beta}.
\end{equation}
A transmitter passing the direct threshold must also pass the recombined threshold. The last expression gives a disk of radius $\big(t(\fade+\kappa\fade')\big)^{1/\beta}$, clipped by the wall. Taking expected counts using Campbell's theorem proves the bounds in expression~\eqref{e.sandwich_general}.
If the mark inequality $\fade'\leq\fade$ holds, the right-hand side of inequality~\eqref{e.pathwise_sandwich} is at most $(1+\kappa)\fade r^{-\beta}$. Equivalently, the inverse powers satisfy $L/(1+\kappa)\leq L^{+}\leq L$. Taking expected counts again proves the bounds in expression~\eqref{e.sandwich}. Figure~\ref{g.pathcounting}(c) illustrates this comparison for common marks.
\end{proof}

The limiting cases agree with the geometry. At wall distance $a=0$, the general upper bound in expression~\eqref{e.sandwich_general} is exact, because the direct and image distances coincide. The simpler upper bound in expression~\eqref{e.sandwich} is then exact under common fading. For reflection coefficient $\kappa=0$, both bounds are exact.

Even under common fading $\fade_i'=\fade_i$, the two sides of the bound in expression~\eqref{e.sandwich} differ asymptotically by the factor $(1+\kappa)^{2/\beta}$. The bounds alone therefore do not give the exact leading coefficient. We next find that coefficient by first studying the unit fading geometry and then allowing general paired marks.

Figure~\ref{g.thresholdregions} shows the geometry behind the next result.
Panel (a) recalls the clipped disk for a direct path. Panel (b) shows the rescaled recombined region and the half-disk it approaches.

\begin{widefigure}
\centering
% Figure source: fig_threshold_regions.tex
% TikZ source for Figure 4 (threshold regions).
% Kept separate from the main article source by design.
% Two columns give the callouts their own space without shrinking the text.
% Native-width TikZ figure. Scale coordinates, not text.
% True TikZ render. Native width 16.1 cm, intended to span both columns.
% Panel (b): beta=3, kappa=1, epsilon=0.35, unit marks.
% Its curved edge is computed from the threshold equation, not hand sketched.
% Local colours are independent of the manuscript revision colours.
\begingroup
\color{black}
\definecolor{thresholdBlue}{RGB}{0,114,178}
\definecolor{thresholdRed}{RGB}{178,45,48}
\begin{tikzpicture}[
  scale=0.6525,x=1cm,y=1cm,
  font=\fontsize{11}{13}\selectfont,
  every node/.style={inner sep=2pt,outer sep=0pt},
  line cap=round,line join=round,
  thresholdwall/.style={draw=black,line width=1.1pt},
  thresholdcallout/.style={draw=black,line width=0.65pt,->,>=stealth},
  thresholdlimit/.style={draw=black,line width=0.85pt,dash pattern=on 4pt off 3pt}
]
\path[use as bounding box] (0,-3.0) rectangle (16.1,3.30);

% (a) Direct-path disk, clipped at x=a. The receiver is the circle centre.
\begin{scope}
  \clip (1.3,-2.2) rectangle (4.45,2.2);
  \fill[thresholdBlue!18] (3.50,0) circle[radius=2.10];
\end{scope}
\draw[black,line width=0.9pt] (3.50,0) circle[radius=2.10];
\draw[thresholdwall] (4.45,-2.48) -- (4.45,2.48);
\draw[thresholdwall] (4.34,-2.48) -- (4.56,-2.48);
\draw[thresholdwall] (4.34,2.48) -- (4.56,2.48);
\node[anchor=south] at (4.45,2.86) {wall};
\node at (2.91,0.91) {$B_o(d)\cap\mathcal L$};
\fill[black] (3.50,0) circle[radius=2.1pt];
\node[anchor=north east] at (3.43,-0.08) {$o$};
\node at (3.50,-2.91) {(a) Direct-path threshold};

% (b) The true recombined threshold region, clipped at u=epsilon.
\path[fill=thresholdRed!12,draw=thresholdRed,line width=1.0pt]
  plot coordinates {
    (12.0333699,2.0173447)
    (12.0053202,2.0172249)
    (11.9774095,2.0168679)
    (11.9496370,2.0162767)
    (11.9220021,2.0154543)
    (11.8945039,2.0144035)
    (11.8671415,2.0131266)
    (11.8399139,2.0116259)
    (11.8128202,2.0099033)
    (11.7858595,2.0079606)
    (11.7590307,2.0057995)
    (11.7323330,2.0034213)
    (11.7057654,2.0008272)
    (11.6793270,1.9980185)
    (11.6530170,1.9949958)
    (11.6268345,1.9917601)
    (11.6007788,1.9883120)
    (11.5748491,1.9846520)
    (11.5490449,1.9807805)
    (11.5233656,1.9766979)
    (11.4978107,1.9724043)
    (11.4723797,1.9679000)
    (11.4470724,1.9631851)
    (11.4218884,1.9582595)
    (11.3968276,1.9531232)
    (11.3718899,1.9477761)
    (11.3470753,1.9422181)
    (11.3223838,1.9364491)
    (11.2978156,1.9304690)
    (11.2733709,1.9242774)
    (11.2490501,1.9178742)
    (11.2248536,1.9112591)
    (11.2007818,1.9044321)
    (11.1768353,1.8973927)
    (11.1530148,1.8901410)
    (11.1293210,1.8826765)
    (11.1057546,1.8749993)
    (11.0823166,1.8671090)
    (11.0590079,1.8590056)
    (11.0358295,1.8506890)
    (11.0127825,1.8421590)
    (10.9898680,1.8334156)
    (10.9670872,1.8244589)
    (10.9444415,1.8152887)
    (10.9219320,1.8059053)
    (10.8995602,1.7963085)
    (10.8773276,1.7864987)
    (10.8552355,1.7764759)
    (10.8332856,1.7662404)
    (10.8114794,1.7557924)
    (10.7898185,1.7451324)
    (10.7683047,1.7342605)
    (10.7469395,1.7231772)
    (10.7257248,1.7118830)
    (10.7046623,1.7003784)
    (10.6837538,1.6886639)
    (10.6630011,1.6767401)
    (10.6424062,1.6646077)
    (10.6219710,1.6522674)
    (10.6016972,1.6397199)
    (10.5815870,1.6269660)
    (10.5616421,1.6140067)
    (10.5418647,1.6008427)
    (10.5222567,1.5874750)
    (10.5028202,1.5739047)
    (10.4835571,1.5601329)
    (10.4644694,1.5461605)
    (10.4455594,1.5319888)
    (10.4268288,1.5176190)
    (10.4082800,1.5030523)
    (10.3899149,1.4882900)
    (10.3717355,1.4733336)
    (10.3537440,1.4581843)
    (10.3359424,1.4428437)
    (10.3183328,1.4273132)
    (10.3009173,1.4115943)
    (10.2836980,1.3956888)
    (10.2666768,1.3795981)
    (10.2498559,1.3633241)
    (10.2332372,1.3468683)
    (10.2168229,1.3302327)
    (10.2006150,1.3134189)
    (10.1846155,1.2964289)
    (10.1688263,1.2792646)
    (10.1532495,1.2619278)
    (10.1378871,1.2444207)
    (10.1227410,1.2267451)
    (10.1078132,1.2089033)
    (10.0931056,1.1908972)
    (10.0786201,1.1727291)
    (10.0643586,1.1544010)
    (10.0503230,1.1359153)
    (10.0365152,1.1172743)
    (10.0229369,1.0984801)
    (10.0095900,1.0795351)
    (9.9964762,1.0604417)
    (9.9835974,1.0412023)
    (9.9709553,1.0218194)
    (9.9585515,1.0022953)
    (9.9463878,0.9826327)
    (9.9344658,0.9628339)
    (9.9227872,0.9429017)
    (9.9113536,0.9228386)
    (9.9001664,0.9026472)
    (9.8892274,0.8823302)
    (9.8785380,0.8618902)
    (9.8680996,0.8413300)
    (9.8579138,0.8206523)
    (9.8479820,0.7998599)
    (9.8383055,0.7789556)
    (9.8288858,0.7579422)
    (9.8197240,0.7368225)
    (9.8108217,0.7155993)
    (9.8021799,0.6942757)
    (9.7937999,0.6728545)
    (9.7856829,0.6513386)
    (9.7778301,0.6297310)
    (9.7702426,0.6080346)
    (9.7629214,0.5862525)
    (9.7558676,0.5643877)
    (9.7490823,0.5424432)
    (9.7425663,0.5204220)
    (9.7363207,0.4983272)
    (9.7303463,0.4761620)
    (9.7246440,0.4539293)
    (9.7192146,0.4316323)
    (9.7140588,0.4092742)
    (9.7091775,0.3868580)
    (9.7045714,0.3643870)
    (9.7002410,0.3418643)
    (9.6961871,0.3192931)
    (9.6924101,0.2966765)
    (9.6889108,0.2740177)
    (9.6856895,0.2513201)
    (9.6827468,0.2285867)
    (9.6800830,0.2058208)
    (9.6776986,0.1830256)
    (9.6755939,0.1602044)
    (9.6737692,0.1373605)
    (9.6722247,0.1144969)
    (9.6709608,0.0916171)
    (9.6699776,0.0687242)
    (9.6692751,0.0458216)
    (9.6688536,0.0229124)
    (9.6687131,0.0000000)
    (9.6688536,-0.0229124)
    (9.6692751,-0.0458216)
    (9.6699776,-0.0687242)
    (9.6709608,-0.0916171)
    (9.6722247,-0.1144969)
    (9.6737692,-0.1373605)
    (9.6755939,-0.1602044)
    (9.6776986,-0.1830256)
    (9.6800830,-0.2058208)
    (9.6827468,-0.2285867)
    (9.6856895,-0.2513201)
    (9.6889108,-0.2740177)
    (9.6924101,-0.2966765)
    (9.6961871,-0.3192931)
    (9.7002410,-0.3418643)
    (9.7045714,-0.3643870)
    (9.7091775,-0.3868580)
    (9.7140588,-0.4092742)
    (9.7192146,-0.4316323)
    (9.7246440,-0.4539293)
    (9.7303463,-0.4761620)
    (9.7363207,-0.4983272)
    (9.7425663,-0.5204220)
    (9.7490823,-0.5424432)
    (9.7558676,-0.5643877)
    (9.7629214,-0.5862525)
    (9.7702426,-0.6080346)
    (9.7778301,-0.6297310)
    (9.7856829,-0.6513386)
    (9.7937999,-0.6728545)
    (9.8021799,-0.6942757)
    (9.8108217,-0.7155993)
    (9.8197240,-0.7368225)
    (9.8288858,-0.7579422)
    (9.8383055,-0.7789556)
    (9.8479820,-0.7998599)
    (9.8579138,-0.8206523)
    (9.8680996,-0.8413300)
    (9.8785380,-0.8618902)
    (9.8892274,-0.8823302)
    (9.9001664,-0.9026472)
    (9.9113536,-0.9228386)
    (9.9227872,-0.9429017)
    (9.9344658,-0.9628339)
    (9.9463878,-0.9826327)
    (9.9585515,-1.0022953)
    (9.9709553,-1.0218194)
    (9.9835974,-1.0412023)
    (9.9964762,-1.0604417)
    (10.0095900,-1.0795351)
    (10.0229369,-1.0984801)
    (10.0365152,-1.1172743)
    (10.0503230,-1.1359153)
    (10.0643586,-1.1544010)
    (10.0786201,-1.1727291)
    (10.0931056,-1.1908972)
    (10.1078132,-1.2089033)
    (10.1227410,-1.2267451)
    (10.1378871,-1.2444207)
    (10.1532495,-1.2619278)
    (10.1688263,-1.2792646)
    (10.1846155,-1.2964289)
    (10.2006150,-1.3134189)
    (10.2168229,-1.3302327)
    (10.2332372,-1.3468683)
    (10.2498559,-1.3633241)
    (10.2666768,-1.3795981)
    (10.2836980,-1.3956888)
    (10.3009173,-1.4115943)
    (10.3183328,-1.4273132)
    (10.3359424,-1.4428437)
    (10.3537440,-1.4581843)
    (10.3717355,-1.4733336)
    (10.3899149,-1.4882900)
    (10.4082800,-1.5030523)
    (10.4268288,-1.5176190)
    (10.4455594,-1.5319888)
    (10.4644694,-1.5461605)
    (10.4835571,-1.5601329)
    (10.5028202,-1.5739047)
    (10.5222567,-1.5874750)
    (10.5418647,-1.6008427)
    (10.5616421,-1.6140067)
    (10.5815870,-1.6269660)
    (10.6016972,-1.6397199)
    (10.6219710,-1.6522674)
    (10.6424062,-1.6646077)
    (10.6630011,-1.6767401)
    (10.6837538,-1.6886639)
    (10.7046623,-1.7003784)
    (10.7257248,-1.7118830)
    (10.7469395,-1.7231772)
    (10.7683047,-1.7342605)
    (10.7898185,-1.7451324)
    (10.8114794,-1.7557924)
    (10.8332856,-1.7662404)
    (10.8552355,-1.7764759)
    (10.8773276,-1.7864987)
    (10.8995602,-1.7963085)
    (10.9219320,-1.8059053)
    (10.9444415,-1.8152887)
    (10.9670872,-1.8244589)
    (10.9898680,-1.8334156)
    (11.0127825,-1.8421590)
    (11.0358295,-1.8506890)
    (11.0590079,-1.8590056)
    (11.0823166,-1.8671090)
    (11.1057546,-1.8749993)
    (11.1293210,-1.8826765)
    (11.1530148,-1.8901410)
    (11.1768353,-1.8973927)
    (11.2007818,-1.9044321)
    (11.2248536,-1.9112591)
    (11.2490501,-1.9178742)
    (11.2733709,-1.9242774)
    (11.2978156,-1.9304690)
    (11.3223838,-1.9364491)
    (11.3470753,-1.9422181)
    (11.3718899,-1.9477761)
    (11.3968276,-1.9531232)
    (11.4218884,-1.9582595)
    (11.4470724,-1.9631851)
    (11.4723797,-1.9679000)
    (11.4978107,-1.9724043)
    (11.5233656,-1.9766979)
    (11.5490449,-1.9807805)
    (11.5748491,-1.9846520)
    (11.6007788,-1.9883120)
    (11.6268345,-1.9917601)
    (11.6530170,-1.9949958)
    (11.6793270,-1.9980185)
    (11.7057654,-2.0008272)
    (11.7323330,-2.0034213)
    (11.7590307,-2.0057995)
    (11.7858595,-2.0079606)
    (11.8128202,-2.0099033)
    (11.8399139,-2.0116259)
    (11.8671415,-2.0131266)
    (11.8945039,-2.0144035)
    (11.9220021,-2.0154543)
    (11.9496370,-2.0162767)
    (11.9774095,-2.0168679)
    (12.0053202,-2.0172249)
    (12.0333699,-2.0173447)
  } -- cycle;

% The limiting half-disk is centred at o and has its diameter at u=0.
% Draw it after the fill so every dashed edge remains visible.
\draw[thresholdlimit] (11.45,2.10)
  arc[start angle=90,end angle=270,radius=2.10] -- cycle;

% The solid wall is the straight right edge of D_epsilon.
\draw[thresholdwall] (12.0333699,-2.48) -- (12.0333699,2.48);
\draw[thresholdwall] (11.9233699,-2.48) -- (12.1433699,-2.48);
\draw[thresholdwall] (11.9233699,2.48) -- (12.1433699,2.48);

% Labels occupy reserved spaces above and inside the diagram.
\node[anchor=south] at (9.46,2.86) {limiting half-disk};
\draw[thresholdcallout] (9.46,2.67) -- (9.8413067,1.3498540);

\node[anchor=south] at (13.78,2.86) {rescaled wall};
\draw[thresholdcallout] (13.23,2.69) -- (12.0333699,2.28);

\node[text=thresholdRed] at (10.56,0.67) {$D_\epsilon$};
% A small white halo keeps the receiver visible on the dashed diameter.
\fill[white] (11.45,0) circle[radius=3.2pt];
\fill[black] (11.45,0) circle[radius=2.1pt];
\node[anchor=north east] at (11.38,-0.08) {$o$};
\node at (11.45,-2.91) {(b) Recombined threshold};
\end{tikzpicture}
\endgroup

\caption{Threshold regions behind the mean measure calculations.
(a) The shaded direct-path region is the disk $B_o(d)$ clipped by the wall.
(b) The shaded recombined region $D_\epsilon$ lies on the transmitter side of the rescaled wall $u=\epsilon$. Its curved edge is computed from the threshold condition in the Appendix, using unit marks, path-loss exponent $\beta=3$, reflection coefficient $\kappa=1$, and rescaled wall distance $\epsilon=0.35$. Although nearly semicircular to the eye, this edge is not generally circular. The dashed outline is the limiting half-disk $D_0$, whose diameter passes through the receiver $o$. As the distance $\epsilon$ tends to zero, the region $D_\epsilon$ approaches $D_0$; the Appendix bounds it between two clipped disks with the same limit.}
\label{g.thresholdregions}
\end{widefigure}

\begin{Proposition}[Exact asymptotics]\label{p.asympt}
Suppose that the transmitter process $\Phi$ has constant density $\lambda>0$, fixed wall distance $a\geq0$, and path-loss exponent $\beta>0$. Let the positive mark pairs $(\fade_i,\fade_i')$ be independent across transmitters, have a common joint distribution, be independent of their locations, and satisfy $\E[\fade^{2/\beta}]+\E[(\fade')^{2/\beta}]<\infty$. The two marks within a pair need not share a distribution or be independent. As the threshold $t$ tends to infinity, the recombined mean measure satisfies
\begin{equation}\label{e.asympt}
\Measure_{+}(t)
\sim\frac{\pi\lambda}{2}\E[(\fade+\kappa\fade')^{2/\beta}]\,t^{2/\beta}.
\end{equation}
The ratio of the recombined and superposed mean measures therefore converges to
\begin{equation}\label{e.ratio_general_marks}
\frac{\Measure_{+}(t)}{\Measure_{\cup}(t)}
\longrightarrow
\frac{\E[(\fade+\kappa\fade')^{2/\beta}]}
{\E[\fade^{2/\beta}]+\kappa^{2/\beta}\E[(\fade')^{2/\beta}]}.
\end{equation}
\end{Proposition}

\begin{proof}[Proof sketch]
Write the rescaled wall distance as $\epsilon=a t^{-1/\beta}$. After rescaling space by the factor $t^{-1/\beta}$, the distance $\epsilon$ tends to zero as the threshold $t$ tends to infinity. For fixed mark values $f,f'>0$, the rescaled direct and image distances then coincide in the limit, so the recombined threshold region approaches the left half of a disk with radius $(f+\kappa f')^{1/\beta}$, as illustrated in Figure~\ref{g.thresholdregions}(b). Bounding the finite-$\epsilon$ region between clipped disks gives convergence of its area. Averaging the limiting areas over the mark pair gives the asymptotic expression~\eqref{e.asympt}. Propositions~\ref{p.los} and~\ref{p.nlos} give the corresponding direct and reflected coefficients at large thresholds; adding them gives the ratio in expression~\eqref{e.ratio_general_marks}. The details are given in the Appendix.
\end{proof}

\begin{Remark}[A universal ratio]\label{r.ratio}
Under common fading with positive marks, the equality $\fade'=\fade$ holds. If the fractional moment $\E[\fade^{2/\beta}]$ is finite, it cancels from expression~\eqref{e.ratio_general_marks}. The limit then depends only on the reflection coefficient and path-loss exponent, not on the fading distribution, wall distance $a$, or transmitter density $\lambda$:
\begin{equation}\label{e.ratio}
\frac{\Measure_{+}(t)}{\Measure_{\cup}(t)}\ \longrightarrow\ \frac{(1+\kappa)^{2/\beta}}{1+\kappa^{2/\beta}}\,,\qquad t\to\infty\,.
\end{equation}
For a lossless wall with reflection coefficient $\kappa=1$, the limit is $2^{2/\beta-1}$. With path-loss exponent $\beta=3$, it is $2^{-1/3}\approx0.794$. By comparison, independent unit-mean exponential marks give $5/6$ for path-loss exponent $\beta=3$ and reflection coefficient $\kappa=1$, showing that the general limit depends on the distribution of the mark pair.

For unit fading, the ratio can exceed one at small thresholds, although the limit under common fading in expression~\eqref{e.ratio} is below one when the path-loss exponent $\beta>2$ and reflection coefficient $\kappa>0$. With a positive wall distance $a>0$ and reflection coefficient $0<\kappa\leq1$, the reflected process $\Xi'$ contributes no points for thresholds $0<t\leq a^{\beta}/\kappa$, yet recombination already enlarges the direct threshold region. Thus the ratio satisfies the inequality
\begin{equation}\label{e.ratioabove}
\frac{\Measure_{+}(t)}{\Measure_{\cup}(t)}>1,\qquad 0<t\leq\frac{a^{\beta}}{\kappa}.
\end{equation}
A reflected path can therefore be too weak to count on its own yet still push its paired transmitter across the recombined threshold. The ratio must therefore cross one under these unit fading assumptions. For fixed paired mark distributions and positive wall distance $a>0$, scaling shows more generally that the ratio depends on the threshold $t$ and wall distance $a$ only through $t/a^{\beta}$ and is independent of the transmitter density $\lambda$. Figure~\ref{g.combratio} illustrates this dependence for unit fading.
\end{Remark}

Figure~\ref{g.combintensity} illustrates the mean measure of the recombined signal process together with the sandwich bounds of Proposition~\ref{p.sandwich}, computed by simulation and by direct numerical integration of expression~\eqref{e.intenscombdef}. Figure~\ref{g.combratio} also shows the ratio in expression~\eqref{e.ratio} converging to its limit; for a lossless wall and path-loss exponent $\beta=3$, the limit is about $0.794$. The threshold applies to inverse received power, so moving right in Figure~\ref{g.combratio} means admitting weaker signals.

The simulations use Poisson transmitters as an example of a constant-density field. Figure~\ref{g.combintensity} uses 400 realizations in a half-disk of radius $10$ centred at the wall point $(a,0)$. The upper sandwich disk contains every plotted threshold region, so no finite-window correction is needed. Numerical integration locates each vertical boundary by the bisection method and integrates the cross-sectional lengths. Figure~\ref{g.combratio} uses rescaled coordinates; the open markers independently recompute the ratios in the original coordinates. Both figures use unit fading. The geometric results hold under their stated assumptions, while the simpler sandwich bound in expression~\eqref{e.sandwich} additionally requires the mark inequality $\fade'\leq\fade$ almost surely.

\begin{figure}[t]
\centering
% Figure source: python/make_single_wall_figures.py
% Numerical checks: python/verify_single_wall.py
{\includegraphics[width=\columnwidth]{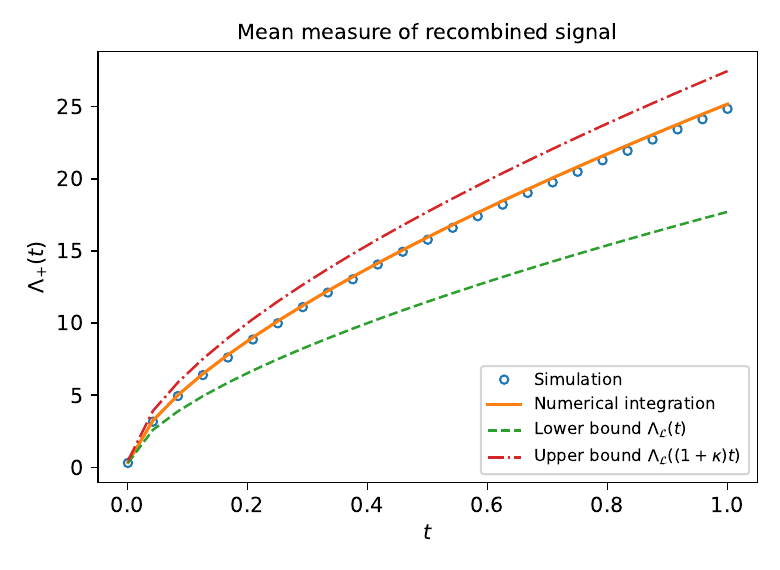}}
\caption{The mean measure $\Measure_{+}(t)$ of the recombined signal process, with the sandwich bounds of Proposition~\ref{p.sandwich}. Parameters: density $\lambda=10$, exponent $\beta=3$, distance $a=0.1$, coefficient $\kappa=1$, unit marks.}
\label{g.combintensity}
\end{figure}

\begin{figure}[t]
\centering
% Figure source: python/make_single_wall_figures.py
% Numerical checks: python/verify_single_wall.py
{\includegraphics[width=\columnwidth]{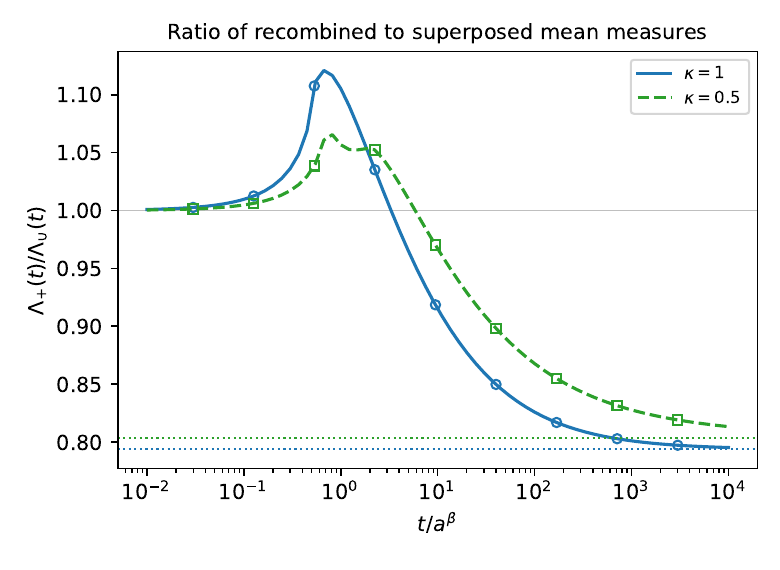}}
\caption{The ratio of the recombined to the superposed mean measures against the rescaled threshold $t/a^{\beta}$, for reflection coefficients $\kappa=1$ and $\kappa=0.5$ and path-loss exponent $\beta=3$. Open markers give independent checks in the original coordinates, using wall distances $a=0.1$ and $a=0.5$ for reflection coefficients $\kappa=1$ and $\kappa=0.5$, respectively. Dotted lines show the limits in expression~\eqref{e.ratio}. Each curve applies to every positive wall distance and crosses one before tending to its limit.}
\label{g.combratio}
\end{figure}

% End Section III.

\section{Poisson transmitter field}\label{s.poisson}
The preceding mean measures, bounds, and asymptotics used only the constant transmitter density, the stated mark assumptions, and Campbell's theorem. We now assume that the transmitter process $\Phi$ is homogeneous Poisson with density $\lambda$. We hold this assumption throughout this section.

\subsection{Mapped processes and strongest received power}\label{s.mapped}
The Poisson assumption now turns the mean measures into full statistical descriptions of the signal processes. Reflection preserves areas, so the spatial image process $\Phi'$ is also homogeneous Poisson with density $\lambda$ on the opposite half-plane $\righthalf$. The two spatial processes are not independent: reflection fixes each image once its transmitter is placed. The following proposition collects the corresponding signal processes and their mean measures.

\begin{Proposition}[Poisson signal processes]\label{p.poisson_signals}
Suppose that the transmitter process $\Phi$ is homogeneous Poisson on the left half-plane $\lefthalf$ with density $\lambda>0$, under the mark and finiteness assumptions of Section~\ref{sec:One_Wall}. The direct propagation process $\Xi$ and reflected propagation process $\Xi'$ are each Poisson on the positive real line $\R^+$. Their respective mean measures $\Measure_{\lefthalf}$ and $\Measure_{\righthalf}$ are given by Propositions~\ref{p.los} and~\ref{p.nlos}.

The recombined propagation process $\Xi_{+}$ is also Poisson. For power-law path loss and general paired marks, its mean measure $\Measure_{+}$ is given by expression~\eqref{e.intenscombdef}. For reflection coefficient $\kappa=0$, the reflected process $\Xi'$ is empty and has mean measure zero, while the recombined process $\Xi_{+}$ equals the direct process $\Xi$.
\end{Proposition}

\begin{proof}
The independent mark pairs give a Poisson process of marked transmitters. For a positive reflection coefficient $\kappa>0$, each of the three signal models maps one marked transmitter to one inverse received power. The marking and mapping theorems therefore give the three Poisson processes~\cite[Secs.~2.3 and~5.2]{kingman1992poisson}, with the mean measures calculated in Sections~\ref{sec:One_Wall} and~\ref{s.both}. The case of zero reflection follows from the definitions.
\end{proof}

The mapped processes $\Xi$, $\Xi'$, and $\Xi_{+}$ are not generally homogeneous. Although each process is marginally Poisson, the superposed process $\Xi_{\cup}=\Xi\cup\Xi'$ is generally not Poisson. For example, when the wall distance is $a=0$, the reflection coefficient is $\kappa=1$, and common fading is used, each direct point is repeated at its reflected counterpart. Counts in the superposition are then even, rather than Poisson distributed. The mean measures still add according to the identity in expression~\eqref{e.intenssep}.

Proposition~\ref{p.poisson_signals} also gives the distribution of the strongest direct received power. Write its reciprocal, the smallest direct inverse received power, as $L_{\min}\triangleq\min_i L_i$. For a threshold $t\geq0$, the usual void-probability argument gives
\begin{equation}\label{e.minLOS}
\Prob(L_{\min}>t)=\Prob\big(\Xi(t)=0\big)=e^{-\Measure_{\lefthalf}(t)}\,.
\end{equation}
The strongest direct received power is therefore $1/L_{\min}$. Under common fading, it is also the strongest of the paths kept separate, by the pointwise comparison in Section~\ref{s.separate}. Without common fading, a reflected path can instead be strongest; we calculate that probability in Section~\ref{s.strongest_path}. The same void-probability argument applies to the recombined process using its mean measure $\Measure_{+}$.

{\color{revisionred}
\subsection{Paired signal counts}\label{s.counts_common}
We first assume common fading on each transmitter's direct and reflected paths, so the marks satisfy $\fade_i'=\fade_i$. A reflected path can then qualify only if its direct counterpart qualifies. Thus, at a fixed threshold, a qualifying transmitter contributes either one direct path or both paths. These are location-dependent marks of the Poisson transmitter process $\Phi$.

Figure~\ref{g.countsplit} shows the one-path and two-path classes used below.

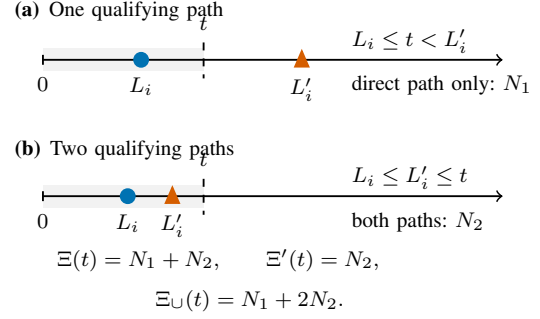
\begin{figure}[t]
\centering
% Figure source: fig_signal_counts.tex
% TikZ source for Figure 7 (paired signal counts).
% Kept separate from the main article source by design.
% Geometry reduced without resizing node text or symbols.
\begingroup
\definecolor{countBlue}{RGB}{0,114,178}
\definecolor{countOrange}{RGB}{213,94,0}
\begin{tikzpicture}[scale=0.82,x=0.72cm,y=1.00cm,font=\footnotesize,
    draw=black,text=black,line width=0.7pt]
\path[use as bounding box] (0,-4.28) rectangle (11.8,1.20);

% Panel (a): direct path only.
\node[anchor=west] at (0.0,0.78) {\textbf{(a)}\hspace{0.45em}One qualifying path};
\fill[black!5] (0.85,-0.18) rectangle (4.45,0.18);
\draw[->] (0.85,0) -- (11.15,0);
\draw (0.85,-0.08) -- (0.85,0.08);
\node[anchor=north] at (0.85,-0.13) {$0$};
\draw[dashed] (4.45,0.28) -- (4.45,-0.30);
\node[anchor=south] at (4.45,0.33) {$t$};
\filldraw[draw=countBlue,fill=countBlue] (3.05,0) circle[radius=3.2pt];
\filldraw[draw=countOrange,fill=countOrange]
    (6.65,0.14) -- (6.50,-0.11) -- (6.80,-0.11) -- cycle;
\node[anchor=north] at (3.05,-0.13) {$L_i$};
\node[anchor=north] at (6.65,-0.13) {$L_i'$};
\node[anchor=west] at (7.55,0.30) {$L_i\leq t<L_i'$};
\node[anchor=west] at (7.55,-0.40) {direct path only: $N_1$};

% Panel (b): both paths.
\node[anchor=west] at (0.0,-1.45) {\textbf{(b)}\hspace{0.45em}Two qualifying paths};
\fill[black!5] (0.85,-2.38) rectangle (4.45,-2.02);
\draw[->] (0.85,-2.20) -- (11.15,-2.20);
\draw (0.85,-2.28) -- (0.85,-2.12);
\node[anchor=north] at (0.85,-2.33) {$0$};
\draw[dashed] (4.45,-1.92) -- (4.45,-2.50);
\node[anchor=south] at (4.45,-1.87) {$t$};
\filldraw[draw=countBlue,fill=countBlue] (2.75,-2.20) circle[radius=3.2pt];
\filldraw[draw=countOrange,fill=countOrange]
    (3.75,-2.06) -- (3.60,-2.31) -- (3.90,-2.31) -- cycle;
\node[anchor=north] at (2.75,-2.33) {$L_i$};
\node[anchor=north] at (3.75,-2.33) {$L_i'$};
\node[anchor=west] at (7.55,-1.90) {$L_i\leq L_i'\leq t$};
\node[anchor=west] at (7.55,-2.62) {both paths: $N_2$};

% Count identities.
\node[anchor=west] at (0.95,-3.23)
    {$\Xi(t)=N_1+N_2,\qquad \Xi'(t)=N_2,$};
\node[anchor=west] at (3.15,-3.88) {$\Xi_{\cup}(t)=N_1+2N_2.$};
% Axis quantity is already identified by the caption.
\end{tikzpicture}
\endgroup
\caption{Signal counting under common fading at a fixed threshold $t$. The shaded intervals mark the interval $(0,t]$. A qualifying transmitter contributes either one direct path or two paths. A reflected-only case cannot occur because the reflected inverse power $L_i'$ is at least the direct inverse power $L_i$. These two cases generate the count decomposition shown below the panels.}
\label{g.countsplit}
\end{figure}

For a fixed threshold $t>0$ on inverse received power, the direct count $\Xi(t)$, reflected count $\Xi'(t)$, superposed count $\Xi_{\cup}(t)$, and recombined count $\Xi_{+}(t)$ record inverse received powers in the interval $(0,t]$. Figure~\ref{g.path_type_probs} illustrates the marking probabilities under unit, common, and independent fading; only independent fading permits reflected-only paths.

\begin{figure}[t]
\centering
\includegraphics[width=\columnwidth]{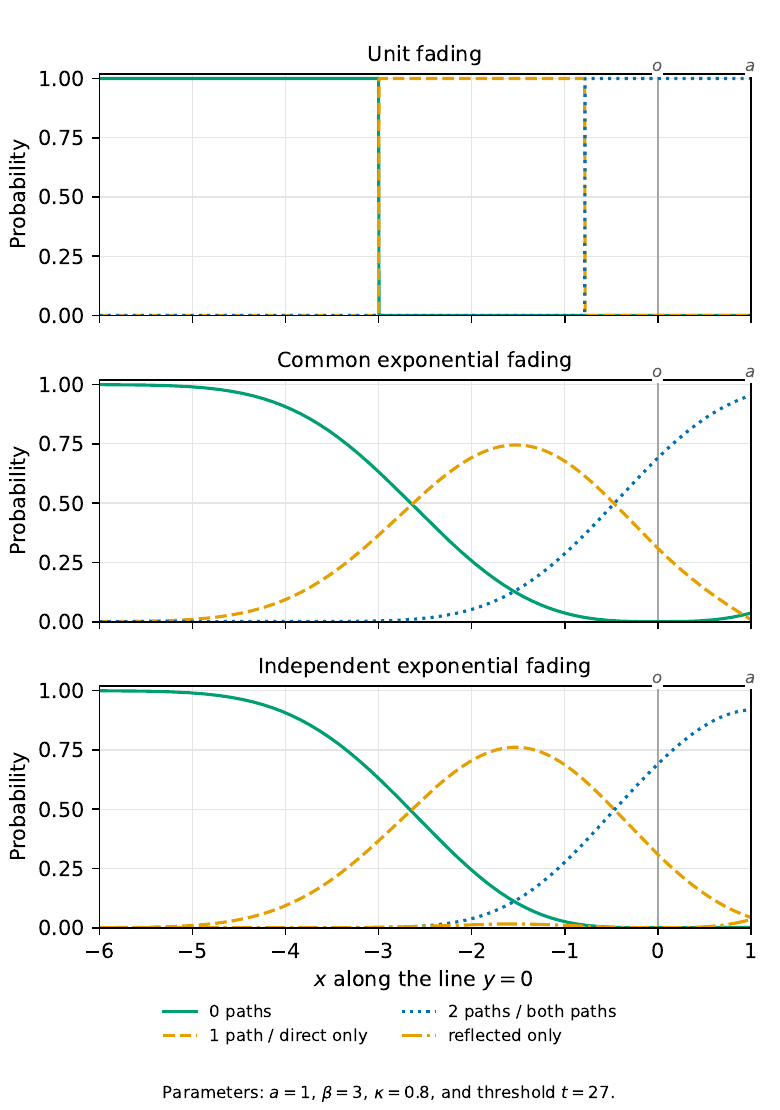}
\caption{Path-type probabilities along the line $y=0$. Top: unit fading. Middle: common exponential fading, with shared mark $F_i'=F_i$. Bottom: independent exponential fading, which permits reflected-only paths. Vertical guides mark the receiver $o$ and wall position $a$. Parameters: wall distance $a=1$, exponent $\beta=3$, coefficient $\kappa=0.8$, and threshold $t=27$.}
\label{g.path_type_probs}
\end{figure}

Under common fading, the next proposition turns the two classes in Figure~\ref{g.countsplit} into independent Poisson counts.

\begin{Proposition}[Paired signal counts]\label{p.paired_counts}
Assume that the transmitter process $\Phi$ is homogeneous Poisson with density $\lambda>0$, that the paths have common fading $\fade_i'=\fade_i$, and that the path-loss, mark, and finiteness conditions of Sections~\ref{sec:One_Wall} and~\ref{s.both} hold. Fix a threshold $t>0$. Let the count $N_1(t)$ record transmitters contributing one qualifying path, and let the count $N_2(t)$ record transmitters contributing two qualifying paths. These counts are independent. The one-path count has the distribution
\begin{equation}\label{e.count_one_path}
N_1(t)\sim\operatorname{Poisson}\!\left(\Measure_{\lefthalf}(t)-\Measure_{\righthalf}(t)\right).
\end{equation}
The two-path count has the distribution
\begin{equation}\label{e.count_two_path}
N_2(t)\sim\operatorname{Poisson}\!\left(\Measure_{\righthalf}(t)\right).
\end{equation}
The direct count satisfies the identity
\begin{equation}\label{e.count_direct_decomposition}
\Xi(t)=N_1(t)+N_2(t).
\end{equation}
The reflected count satisfies the identity
\begin{equation}\label{e.count_reflected_decomposition}
\Xi'(t)=N_2(t).
\end{equation}
The superposed count satisfies the identity
\begin{equation}\label{e.count_union_decomposition}
\Xi_{\cup}(t)=N_1(t)+2N_2(t).
\end{equation}
Consequently, the covariance of the direct and reflected counts is
\begin{equation}\label{e.count_covariance}
\operatorname{Cov}\!\left(\Xi(t),\Xi'(t)\right)=\Measure_{\righthalf}(t).
\end{equation}
The variance of the superposed count is
\begin{equation}\label{e.count_variance}
\Var\!\left(\Xi_{\cup}(t)\right)=\Measure_{\cup}(t)+2\Measure_{\righthalf}(t).
\end{equation}
\end{Proposition}

\begin{proof}
Under common fading, the path order inequality~\eqref{e.path_power_order} makes reflected qualification a subset of direct qualification. The expression~\eqref{e.mean_reflected_general} for the reflected mean measure gives the mean $\Measure_{\righthalf}(t)$ of the two-path count $N_2(t)$. The expression~\eqref{e.mean_direct_general} for the direct mean measure gives the mean $\Measure_{\lefthalf}(t)$, so subtracting the two-path mean gives the one-path mean in expression~\eqref{e.count_one_path}. The zero-path, one-path, and two-path classes partition the marked Poisson transmitter process $\Phi$. The location-dependent marking theorem makes these classes independent Poisson processes~\cite[Sec.~5.2, pp.~55--56]{kingman1992poisson}, which proves the distributions in expressions~\eqref{e.count_one_path} and~\eqref{e.count_two_path}. The count identities in expressions~\eqref{e.count_direct_decomposition}--\eqref{e.count_union_decomposition} then give the covariance and variance.
\end{proof}
\FloatBarrier

Thus common fading adds the excess variance $2\Measure_{\righthalf}(t)$ above a Poisson count with the same mean. Since every qualifying reflected path has a qualifying direct counterpart, $\Measure_{\righthalf}(t)\leq\Measure_{\lefthalf}(t)$, so the variance-to-mean ratio lies between one and two. The lower value is attained for reflection coefficient $\kappa=0$; the upper value is attained when a lossless wall passes through the receiver.

For general positive paired marks, the same marking argument still gives independent one-path and two-path counts, although a one-path transmitter may now contribute either path. Let the mean number of transmitters for which both paths qualify be $m_2(t)$. Then
$\Var(\Xi_{\cup}(t))=\Measure_{\cup}(t)+2m_2(t)\leq2\Measure_{\cup}(t)$.
This bound is enough to connect the mean measure asymptotics to observed counts. As the threshold grows, each count divided by its mean becomes arbitrarily close to one with probability tending to one. We denote this convergence in probability by $\xrightarrow{\mathsf P}$.

\begin{Corollary}[Observed count ratio]\label{c.observed_ratio}
Under the assumptions of Proposition~\ref{p.asympt}, suppose further that the transmitter process $\Phi$ is homogeneous Poisson. Then
\begin{equation}\label{e.observed_ratio}
\frac{\Xi_{+}(t)}{\Xi_{\cup}(t)}
\xrightarrow{\mathsf P}
\frac{\E[(\fade+\kappa\fade')^{2/\beta}]}
{\E[\fade^{2/\beta}]+\kappa^{2/\beta}\E[(\fade')^{2/\beta}]},
\qquad t\to\infty.
\end{equation}
Set the ratio to zero when the denominator is zero. Under common fading, the limit reduces to the universal constant in expression~\eqref{e.ratio}.
\end{Corollary}

\begin{proof}
The recombined count $\Xi_{+}(t)$ is Poisson, so its variance equals $\Measure_{+}(t)$. The general marking argument above gives $\Var(\Xi_{\cup}(t))\leq2\Measure_{\cup}(t)$ for the superposed count $\Xi_{\cup}(t)$. Both mean counts grow in proportion to $t^{2/\beta}$, so Chebyshev's inequality makes each count divided by its mean converge in probability to one. Combining these limits with expression~\eqref{e.ratio_general_marks} proves the convergence in expression~\eqref{e.observed_ratio}; the probability that $\Xi_{\cup}(t)=0$ tends to zero.
\end{proof}

\subsection{The strongest individual path}\label{s.strongest_path}
Under common fading, the mark equality $\fade_i'=\fade_i$ makes the direct path of every transmitter at least as strong as its reflected path. With general paired marks, the strongest individual path in the network may instead be reflected. The aggregate received power is a standard shot noise, formed by adding transmitter contributions. Here we instead consider the largest individual power $M=\sup_i\max(1/L_i,1/L_i')$, an example of extremal shot noise that takes their maximum~\cite[Vol.~I, Sec.~2.4]{BB09}. We also identify whether that maximum comes from a direct or reflected path, rather than finding the strongest \emph{recombined} transmitter.

For a transmitter at location $z=(x,y)$, let the direct and reflected inverse powers, using a generic mark pair $(\fade,\fade')$, be
\begin{equation}\label{e.strongest_path_pair}
L(z)\triangleq\frac{r^\beta}{\fade},\qquad
L'(z)\triangleq\frac{\hat r_1^\beta}{\kappa\fade'}.
\end{equation}
For reflection coefficient $\kappa=0$, take the reflected inverse power to be $L'(z)=\infty$. We assign a path to the transmitter according to the smaller of the inverse powers $L(z)$ and $L'(z)$; a tie is assigned to the direct path. For a threshold $t>0$, define the intensity measure $H(t)$ on the interval $(0,t]$ for the Poisson process formed by retaining the stronger path from each transmitter. Define the direct-winner intensity measure $H_D(t)$ on the same interval. Thus the total intensity measure is $H$, while the subscript $D$ selects its direct component; the reflected-winner component has intensity measure $H(t)-H_D(t)$. The two intensity measures $H(t)$ and $H_D(t)$ are given by the integrals
\begin{align}
H(t)&\triangleq\lambda\int_{\lefthalf}
  \Prob\bigl(\min(L(z),L'(z))\leq t\bigr)\,dz,
  \label{e.strongest_path_mean}\\
H_D(t)&\triangleq\lambda\int_{\lefthalf}
  \Prob\bigl(L(z)\leq t,\ L(z)\leq L'(z)\bigr)\,dz.
  \label{e.strongest_path_direct_mean}
\end{align}

When the wall passes through the receiver, the spatial integrals in expressions~\eqref{e.strongest_path_mean} and~\eqref{e.strongest_path_direct_mean} again reduce to half-disk areas. The next lemma evaluates them using the marked projection process calculation behind propagation invariance and network equivalence~\cite[Sec.~4.2.5]{book2018stochastic} and~\cite{netequivalence}.

\begin{Lemma}[Intensity measures for the stronger path at the wall]\label{l.strongest_path_measures}
Assume that the wall distance is $a=0$. The total intensity measure for the stronger path is
\begin{equation}\label{e.strongest_path_mean_at_wall}
H(t)=\frac{\pi\lambda}{2}t^{2/\beta}
\E\!\left[\max(\fade,\kappa\fade')^{2/\beta}\right].
\end{equation}
The intensity measure for direct winners is
\begin{equation}\label{e.strongest_path_direct_mean_at_wall}
H_D(t)=\frac{\pi\lambda}{2}t^{2/\beta}
\E\!\left[\fade^{2/\beta}\Ind(\fade\geq\kappa\fade')\right].
\end{equation}
\end{Lemma}

\begin{proof}
At wall distance $a=0$, the two path lengths from a transmitter equal its distance $r$ from the receiver. For fixed marks, the spatial integral in expression~\eqref{e.strongest_path_mean} selects a half-disk of radius $(t\max(\fade,\kappa\fade'))^{1/\beta}$. The spatial integral in expression~\eqref{e.strongest_path_direct_mean} selects a half-disk of radius $(t\fade)^{1/\beta}$ on the event $\fade\geq\kappa\fade'$. Multiplying the two half-disk areas by the transmitter density $\lambda$ and averaging over the mark pair gives the two intensity measures.
\end{proof}

These measures now give the probability that the winning path is direct and the distribution of the strongest power.

\begin{Proposition}[Strongest individual path]\label{p.strongest_path}
Assume that the transmitter process $\Phi$ is homogeneous Poisson with density $\lambda>0$, path-loss exponent $\beta>0$, and positive mark pairs satisfying the assumptions of Proposition~\ref{p.asympt}. The probability $p_D$ that the strongest individual path is direct, resolving exact ties in favour of the direct path, is
\begin{equation}\label{e.strongest_path_probability}
p_D=\int_0^\infty e^{-H(t)}H_D'(t)\,dt.
\end{equation}
If the wall passes through the receiver, $a=0$, the probability simplifies to
\begin{equation}\label{e.strongest_path_at_wall}
p_D=\frac{\E\!\left[\fade^{2/\beta}\,\Ind(\fade\geq\kappa\fade')\right]}
{\E\!\left[\max(\fade,\kappa\fade')^{2/\beta}\right]}.
\end{equation}
At wall distance $a=0$, the strongest individual received power $M$ has, for a received power level $s>0$, the Fr\'echet distribution
\begin{equation}\label{e.strongest_path_frechet}
\Prob(M\leq s)=\exp\!\left[-\frac{\pi\lambda}{2}
\E\!\left[\max(\fade,\kappa\fade')^{2/\beta}\right]s^{-2/\beta}\right].
\end{equation}
\end{Proposition}

\begin{proof}
Keeping the stronger path from each independently marked Poisson transmitter maps the marked transmitter process $\Phi$ to a Poisson point process of inverse received powers. Marking each mapped point by its winning path splits this process into independent direct-winner and reflected-winner Poisson subprocesses, by the marking theorem~\cite[Sec.~5.2]{kingman1992poisson}. Their total intensity measure on the interval $(0,t]$ is $H(t)$, while the direct-winner intensity measure is $H_D(t)$. The Poisson void probability for the interval $(0,t]$ is $e^{-H(t)}$. Multiplying this probability by the direct-winner intensity increment $H_D'(t)\,dt$ over a short interval and integrating gives the probability in expression~\eqref{e.strongest_path_probability}.

At wall distance $a=0$, substitution of the intensity measures in expressions~\eqref{e.strongest_path_mean_at_wall} and~\eqref{e.strongest_path_direct_mean_at_wall} into expression~\eqref{e.strongest_path_probability} gives the direct-winner probability in expression~\eqref{e.strongest_path_at_wall}. The same Poisson void probability gives the Fr\'echet distribution in expression~\eqref{e.strongest_path_frechet}.
\end{proof}

The Fr\'echet form is the half-plane analogue of the maximum received signal strength distribution for a full-plane Poisson transmitter field~\cite{nguyen2010best}; here the direct and reflected powers remain paired, and we also identify which path supplies the maximum. The exponential factor in the distribution in expression~\eqref{e.strongest_path_frechet} comes from the Poisson void probability, not from an assumption of exponential fading for the marks $\fade_i$ or $\fade_i'$. Independent exponential marks give the following closed form.

\refstepcounter{Theorem}
\noindent\textbf{Corollary~\theTheorem{}\hspace{0.25em}(Independent exponential fading).}\label{c.strongest_path_exponential}
\textit{Assume Proposition~\ref{p.strongest_path} with wall distance $a=0$ and independent unit-mean exponential marks $\fade,\fade'$. For a positive reflection coefficient $\kappa>0$, define the quantity $q=\kappa/(1+\kappa)$. Then the direct-winner probability $p_D$ is}
\begin{equation}\label{e.strongest_path_exponential}
p_D=\frac{1-q^{1+2/\beta}}
{1-q^{1+2/\beta}+\kappa^{2/\beta}
 [1-(1-q)^{1+2/\beta}]}.
\end{equation}
\textit{For a zero reflection coefficient $\kappa=0$, the probability is $p_D=1$.}

\noindent\textit{Proof.} For a positive reflection coefficient $\kappa>0$, conditioning on the direct mark $\fade$ gives the direct-winner moment
\begin{align}
&\E\!\left[\fade^{2/\beta}\Ind(\fade\geq\kappa\fade')\right] \nonumber\\
&\quad=\int_0^\infty x^{2/\beta}e^{-x}\bigl(1-e^{-x/\kappa}\bigr)\,dx \nonumber\\
&\quad=\Gamma(1+2/\beta)\bigl[1-q^{1+2/\beta}\bigr].
\label{e.direct_winner_exponential_moment}
\end{align}
Conditioning instead on the reflected mark $\fade'$ gives the reflected-winner moment
\begin{align}
&\E\!\left[(\kappa\fade')^{2/\beta}\Ind(\kappa\fade'>\fade)\right] \nonumber\\
&\quad=\kappa^{2/\beta}\int_0^\infty x^{2/\beta}e^{-x}\bigl(1-e^{-\kappa x}\bigr)\,dx \nonumber\\
&\quad=\Gamma(1+2/\beta)\kappa^{2/\beta}
   \bigl[1-(1-q)^{1+2/\beta}\bigr].
\label{e.reflected_winner_exponential_moment}
\end{align}
The denominator in expression~\eqref{e.strongest_path_at_wall} is the sum of the moments in expressions~\eqref{e.direct_winner_exponential_moment} and~\eqref{e.reflected_winner_exponential_moment}. Their common Gamma factor cancels, proving expression~\eqref{e.strongest_path_exponential}. The case of zero reflection follows by continuity.

For path-loss exponent $\beta=3$ and reflection coefficient $\kappa=0.8$, the probability is $p_D\approx0.579$ at wall distance $a=0$, compared with one under common fading; with the same parameters and marks and dimensionless product $\lambda a^2=1$, numerical evaluation of expression~\eqref{e.strongest_path_probability} gives $p_D\approx0.991$. Thus positive wall distance strongly favours the direct path despite independent fading.
}

\section{Conclusion}\label{sec:conclusion}
A fixed reflecting wall pairs every real transmitter with a mirror image. Keeping the two received powers apart superposes their signal processes, while adding them transmitter by transmitter produces a recombined process. The corresponding mean measures follow from threshold regions in the plane, without assuming Poisson transmitters. For power-law path loss and general paired fading marks, the ratio of recombined to superposed mean measures at large thresholds is independent of transmitter density and wall distance, but generally depends on the joint mark distribution. Under common fading, this dependence disappears, leaving a constant determined only by the reflection coefficient and path-loss exponent.

For Poisson transmitters, location-dependent marking gives paired count fluctuations and convergence of the observed ratio. Common fading creates excess variance in the superposed count. For general paired marks, a reflected path can be strongest, so we calculate the probability that the winning path is direct. When the wall passes through the receiver, this probability has a closed form and the strongest received power has an exact Fr\'echet distribution. This fixed specular wall does not model programmable surfaces, but it isolates the dependence created by mirror geometry and provides a reference case for reflected wireless signals.

\appendix[Proof of Proposition~\ref{p.asympt}]\label{app:asympt_proof}
{\color{revisionred}First take unit marks $\fade=\fade'=1$. Rescale transmitter location $z=(x,y)$ to coordinates $(u,v)=t^{-1/\beta}(x,y)$ and wall distance $a$ to rescaled distance $\epsilon=a t^{-1/\beta}$. The wall becomes the line $u=\epsilon$, and transmitter coordinate $u\leq\epsilon$, as shown in Figure~\ref{g.thresholdregions}(b). Write the rescaled direct distance as $\bar r(u,v)\triangleq\sqrt{u^2+v^2}$ and the rescaled image distance as $\bar r_1(u,v;\epsilon)\triangleq\sqrt{(2\epsilon-u)^2+v^2}$. These are the scaled distances $t^{-1/\beta}r$ and $t^{-1/\beta}\hat r_1$. The semicolon separates the spatial variables from the rescaled wall distance. We suppress the arguments below when no ambiguity arises. Define the threshold region by
\begin{equation}\label{e.Depsilon}
\begin{aligned}
D_\epsilon\triangleq\{(u,v)\in\R^2:\;&u\leq\epsilon,\\
&\bar r(u,v)^{-\beta}\\
&\quad+\kappa\bar r_1(u,v;\epsilon)^{-\beta}\geq1\}.
\end{aligned}
\end{equation}
The area $|D_\epsilon|$ gives
$\Measure_{+}(t)=\lambda t^{2/\beta}|D_\epsilon|$.

When the rescaled wall distance is zero, $\epsilon=0$, the rescaled direct and image distances coincide:
\begin{equation}\label{e.limit_image_distance}
\bar r_1(u,v;0)=\bar r(u,v).
\end{equation}
Hence the defining inequality in expression~\eqref{e.Depsilon} becomes
\begin{equation}\label{e.limit_region_inequality}
(1+\kappa)\bar r(u,v)^{-\beta}\geq1,
\end{equation}
which is equivalent to
\begin{equation}\label{e.limit_region_disk}
u^2+v^2\leq(1+\kappa)^{2/\beta},
\end{equation}
which defines a disk. Together with the half-plane condition $u\leq0$, this shows that $D_0$ is exactly the left half of a disk with radius
$\rho=(1+\kappa)^{1/\beta}$. Here the value $\epsilon=0$ describes the limiting rescaled geometry, not necessarily a wall at physical distance $a=0$. For every fixed wall distance $a$, the rescaled distance $\epsilon=a t^{-1/\beta}$ tends to zero as $t\to\infty$.

For positive rescaled wall distance $\epsilon>0$ and transmitter coordinate $u\leq\epsilon$, the two rescaled distances satisfy
\begin{equation}\label{e.rescaled_distance_bounds}
\bar r\leq\bar r_1\leq\bar r+2\epsilon.
\end{equation}
Indeed,
\begin{equation}\label{e.rescaled_distance_identity}
\bar r_1^2-\bar r^2
=4\epsilon(\epsilon-u)\geq0,
\end{equation}
and the upper bound follows from the triangle inequality. 
If point $(u,v)$ lies in region $D_\epsilon$, then
\begin{equation}\label{e.rescaled_region_bound}
1\leq \bar r^{-\beta}+\kappa\bar r_1^{-\beta}
\leq(1+\kappa)\bar r^{-\beta},
\end{equation}
so the direct-distance bound $\bar r\leq\rho$ holds. Conversely, for small rescaled distance $\epsilon$, coordinate condition $u\leq\epsilon$ and direct-distance bound $\bar r\leq\rho-2\epsilon$ imply image-distance bound $\bar r_1\leq\rho$, and therefore
\begin{equation}\label{e.rescaled_region_inner}
\bar r^{-\beta}+\kappa\bar r_1^{-\beta}
\geq(1+\kappa)\rho^{-\beta}=1.
\end{equation}
The following inner and outer sets formalize the convergence illustrated in Figure~\ref{g.thresholdregions}(b):
\begin{equation}\label{e.region_inclusions}
B_o(\rho-2\epsilon)\cap\{u\leq\epsilon\}
\subseteq D_\epsilon
\subseteq B_o(\rho)\cap\{u\leq\epsilon\}.
\end{equation}
The inner set (on the left) contains the left half-disk of radius $\rho-2\epsilon$. The outer set (on the right) is contained in the left half-disk of radius $\rho$ together with a vertical strip of width $\epsilon$ and height at most $2\rho$. Hence the bounds on the area $|D_\epsilon|$, namely,
\begin{equation}\label{e.area_bounds}
\frac{\pi}{2}(\rho-2\epsilon)^2
\leq |D_\epsilon|
\leq \frac{\pi}{2}\rho^2+2\rho\epsilon.
\end{equation}
As rescaled distance $\epsilon\to0$, both bounds equal $\frac{\pi}{2}\rho^2+O(\epsilon)$, and therefore
\begin{equation}\label{e.area_asymptotic}
|D_\epsilon|=\frac{\pi}{2}(1+\kappa)^{2/\beta}+O(\epsilon).
\end{equation}
Since the mean measure satisfies $\Measure_{+}(t)=\lambda t^{2/\beta}|D_\epsilon|$ and the rescaled wall distance satisfies $\epsilon=a t^{-1/\beta}\to0$, we obtain the unit fading case of the asymptotic formula in expression~\eqref{e.asympt}.}

{\color{revisionred}For general fixed mark values $f,f'>0$, the rescaled threshold condition becomes
\begin{equation}\label{e.general_rescaled_region}
f\bar r^{-\beta}+\kappa f'\bar r_1^{-\beta}\geq1,
\qquad u\leq\epsilon.
\end{equation}
As the rescaled wall distance $\epsilon$ tends to zero, the rescaled direct and image distances $\bar r$ and $\bar r_1$ coincide, and the qualifying region approaches the left half-disk of radius $(f+\kappa f')^{1/\beta}$. Moreover, the distance inequality $\bar r_1\geq\bar r$ confines the region to the disk of that radius, whose area is $\pi(f+\kappa f')^{2/\beta}$. The moment assumptions make this area integrable. Averaging the limiting half-disk areas by dominated convergence proves the asymptotic formula in expression~\eqref{e.asympt}.

Propositions~\ref{p.los} and~\ref{p.nlos} likewise give the respective coefficients at large thresholds $\pi\lambda\E[\fade^{2/\beta}]/2$ and $\pi\lambda\kappa^{2/\beta}\E[(\fade')^{2/\beta}]/2$ for the direct and reflected mean measures. Their sum proves the ratio in expression~\eqref{e.ratio_general_marks}. No independence or common marginal distribution is required between marks $\fade$ and $\fade'$ within a pair.}

\bibliographystyle{IEEEtran}
\bibliography{./walls,IEEEabrv}

\begin{thebibliography}{10}
\providecommand{\url}[1]{#1}
\csname url@samestyle\endcsname
\providecommand{\newblock}{\relax}
\providecommand{\bibinfo}[2]{#2}
\providecommand{\BIBentrySTDinterwordspacing}{\spaceskip=0pt\relax}
\providecommand{\BIBentryALTinterwordstretchfactor}{4}
\providecommand{\BIBentryALTinterwordspacing}{\spaceskip=\fontdimen2\font plus
\BIBentryALTinterwordstretchfactor\fontdimen3\font minus
  \fontdimen4\font\relax}
\providecommand{\BIBforeignlanguage}[2]{{%
\expandafter\ifx\csname l@#1\endcsname\relax
\typeout{** WARNING: IEEEtran.bst: No hyphenation pattern has been}%
\typeout{** loaded for the language `#1'. Using the pattern for}%
\typeout{** the default language instead.}%
\else
\language=\csname l@#1\endcsname
\fi
#2}}
\providecommand{\BIBdecl}{\relax}
\BIBdecl

\bibitem{keeler2026reflected}
H.~P. Keeler, ``Reflected wireless signals under random spatial sampling,''
  \emph{arXiv preprint arXiv:2601.20699}, 2026.

\bibitem{baccelli1997stochastic}
F.~Baccelli, M.~Klein, M.~Lebourges, and S.~Zuyev, ``Stochastic geometry and
  architecture of communication networks,'' \emph{Telecommunication Systems},
  vol.~7, no.~1, pp. 209--227, 1997.

\bibitem{baccelli2001coverage}
F.~Baccelli and B.~B{\l}aszczyszyn, ``On a coverage process ranging from the
  {Boolean} model to the {Poisson--Voronoi} tessellation with applications to
  wireless communications,'' \emph{Advances in Applied Probability}, vol.~33,
  no.~2, pp. 293--323, 2001.

\bibitem{andrews2011tractable}
J.~G. Andrews, F.~Baccelli, and R.~K. Ganti, ``A tractable approach to coverage
  and rate in cellular networks,'' \emph{IEEE Transactions on Communications},
  vol.~59, no.~11, pp. 3122--3134, 2011.

\bibitem{BB09}
F.~Baccelli and B.~B{\l}aszczyszyn, \emph{Stochastic Geometry and Wireless
  Networks, Volume {I}: Theory}, ser. Foundations and Trends in
  Networking.\hskip 1em plus 0.5em minus 0.4em\relax Now Publishers, 2009,
  vol.~3.

\bibitem{blaszczyszyn2020coverage}
B.~B{\l}aszczyszyn, A.~Brochard, and H.~P. Keeler, ``Coverage probability in
  wireless networks with determinantal scheduling,'' in \emph{Proc. WiOpt},
  2020, pp. 1--8.

\bibitem{kishk2020exploiting}
M.~A. Kishk and M.-S. Alouini, ``Exploiting randomly located blockages for
  large-scale deployment of intelligent surfaces,'' \emph{IEEE Journal on
  Selected Areas in Communications}, vol.~39, no.~4, pp. 1043--1056, 2021.

\bibitem{zhang2021reconfigurable}
C.~Zhang, W.~Yi, Y.~Liu, K.~Yang, and Z.~Ding, ``Reconfigurable intelligent
  surfaces aided multi-cell {NOMA} networks: A stochastic geometry model,''
  \emph{IEEE Transactions on Communications}, vol.~70, no.~2, pp. 951--966,
  2022.

\bibitem{shafique2022stochastic}
T.~Shafique, H.~Tabassum, and E.~Hossain, ``Stochastic geometry analysis of
  {IRS}-assisted downlink cellular networks,'' \emph{IEEE Transactions on
  Communications}, vol.~70, no.~2, pp. 1442--1456, 2022.

\bibitem{pradhan2025beamshaping}
A.~Pradhan, M.~F. Imani, and H.~S. Dhillon, ``A beamshaping framework for
  physically consistent reconfigurable intelligent surfaces,'' in \emph{2025
  IEEE International Conference on Communications (ICC)}, 2025, pp. 6191--6196.

\bibitem{shi2023coverage}
X.~Shi, N.~Deng, N.~Zhao, and D.~Niyato, ``Coverage enhancement in
  millimeter-wave cellular networks via distributed {IRSs},'' \emph{IEEE
  Transactions on Communications}, vol.~71, no.~2, pp. 1153--1167, Feb. 2023.

\bibitem{lin2025ternary}
H.-C. Lin and Q.-Y. Yu, ``Ternary stochastic geometry theory for performance
  analysis of {RIS}-assisted {UDN},'' \emph{IEEE Transactions on
  Communications}, vol.~73, no.~1, pp. 600--613, Jan 2025.

\bibitem{choi2025stochastic}
C.-S. Choi, J.~Kim, and J.~Choi, ``Stochastic geometry analysis of
  {RIS}-assisted cellular networks with reflective intelligent surfaces on
  roads,'' \emph{IEEE Transactions on Communications}, vol.~73, no.~2, pp.
  859--873, Feb 2025.

\bibitem{sun2026stochastic}
G.~Sun, F.~Baccelli, K.~Feng, L.~U. Garcia, and S.~Paris, ``A stochastic
  geometry framework for performance analysis of {RIS}-assisted {OFDM} cellular
  networks,'' \emph{IEEE Transactions on Wireless Communications}, vol.~25, pp.
  5859--5875, 2026.

\bibitem{direnzo2019reflection}
M.~Di~Renzo and J.~Song, ``Reflection probability in wireless networks with
  metasurface-coated environmental objects: An approach based on random spatial
  processes,'' \emph{EURASIP Journal on Wireless Communications and
  Networking}, vol. 2019, no.~1, p.~99, 2019.

\bibitem{narayanan2017coverage}
A.~Narayanan, S.~{T. V.}, and R.~K. Ganti, ``Coverage analysis in millimeter
  wave cellular networks with reflections,'' in \emph{2017 IEEE Global
  Communications Conference (GLOBECOM)}, 2017, pp. 1--6.

\bibitem{book2018stochastic}
B.~B{\l}aszczyszyn, M.~Haenggi, P.~Keeler, and S.~Mukherjee, \emph{Stochastic
  geometry analysis of cellular networks}.\hskip 1em plus 0.5em minus
  0.4em\relax Cambridge University Press, 2018.

\bibitem{keeler2016stronger}
H.~P. Keeler, N.~Ross, A.~Xia, and B.~B{\l}aszczyszyn, ``Stronger wireless
  signals appear more {P}oisson,'' \emph{IEEE Wireless Communications Letters},
  vol.~5, no.~6, pp. 572--575, 2016.

\bibitem{keeler2018wireless}
H.~P. Keeler, N.~Ross, and A.~Xia, ``When do wireless network signals appear
  {P}oisson?'' \emph{Bernoulli}, vol.~24, no.~3, pp. 1973--1994, 2018.

\bibitem{nguyen2010best}
V.~M. Nguyen and F.~Baccelli, ``A stochastic geometry model for the best signal
  quality in a wireless network,'' in \emph{Proc. WiOpt}, May 2010, pp.
  465--471.

\bibitem{kingman1992poisson}
J.~F.~C. Kingman, \emph{Poisson Processes}, 1st~ed.\hskip 1em plus 0.5em minus
  0.4em\relax Oxford University Press, 1993.

\bibitem{lawrence1972catalog}
J.~D. Lawrence, \emph{A Catalog of Special Plane Curves}.\hskip 1em plus 0.5em
  minus 0.4em\relax Dover Publications, 1972.

\bibitem{netequivalence}
B.~B{\l}aszczyszyn and H.~P. Keeler, ``Equivalence and comparison of
  heterogeneous cellular networks,'' in \emph{Proc. of PIMRC/WDN-CN}, 2013.

\end{thebibliography}
\end{document}